%% file: Revised_manuscript.tex
\documentclass[journal]{IEEEtran}
\include{Definitions}
\usepackage{mdframed}
\usepackage{lipsum} 
\usepackage[ruled,vlined,linesnumbered]{algorithm2e}

\definecolor{darkgreen}{RGB}{0,100,0}
\usepackage{etoolbox}

\makeatletter
\newif\ifhb@open
\hb@openfalse

\let\hb@orig@bibitem\@bibitem
\def\@bibitem#1{%
  \ifhb@open\endgroup\fi
  \hb@opentrue
  \begingroup
  \ifcsdef{hb@#1}{\color{\csuse{hb@#1}}}{}%
  \hb@orig@bibitem{#1}\ignorespaces
}

\let\hb@orig@lbibitem\@lbibitem
\def\@lbibitem[#1]#2{%
  \ifhb@open\endgroup\fi
  \hb@opentrue
  \begingroup
  \ifcsdef{hb@#2}{\color{\csuse{hb@#2}}}{}%
  \hb@orig@lbibitem[#1]{#2}\ignorespaces
}

\pretocmd{\endthebibliography}{\ifhb@open\endgroup\hb@openfalse\fi}{}{}
\makeatother

\newcommand{\Nhan}[1]{\textcolor{red}{\footnotesize{\textsf {[Nhan: #1]}}}}

\begin{document}
\setlength{\abovedisplayskip}{3pt}\setlength{\belowdisplayskip}{3pt}
\title{Performance Analysis and Joint Beamforming for Hybrid RIS-Aided Massive MIMO ISAC}
\author{\IEEEauthorblockN{Smriti~Uniyal, \textit{Graduate Student Member, IEEE,} Tianyu~Fang, \textit{Graduate Student Member, IEEE,} \\ Marco Di Renzo, \textit{Fellow, IEEE,} and Markku Juntti, \textit{Fellow, IEEE}}, Nhan~Thanh~Nguyen, \textit{Member, IEEE} 
\thanks{A short version of this paper was presented at the
IEEE WCNC, 2025, [DOI: https://doi.org/10.1109/WCNC61545.2025.10978368] \cite{uniyal_wcnc}.}

\thanks{Smriti Uniyal, Tianyu Fang, Markku Juntti, and  Nhan Thanh Nguyen are with CWC, University of Oulu, Finland. 
Marco Di Renzo is with Universit\'e Paris-Saclay, CNRS, CentraleSup\'elec, Laboratoire des Signaux et Syst\`emes, 3 Rue Joliot-Curie, 91192 Gif-sur-Yvette, France,
and with King's College London, 
CTR-Department of Engineering, WC2R 2LS London, UK. 
}
\vspace*{-0.055cm}}

\maketitle
\begin{abstract}
In integrated sensing and communication (ISAC) systems, stringent sensing performance constraints can severely limit the power available for communication. Hybrid reconfigurable intelligent surfaces (HRISs) with capabilities of both passive reflection and active signal amplification can significantly improve communication performance 
in the power-limited regime.
This motivates us to analyze and optimize the
performance of an HRIS-aided multiple-input-multiple-output (mMIMO) ISAC system. 
We
first  
estimate the effective uplink/downlink channels
using
the minimum mean square error method. 
We then derive closed-form expressions for the communication sum-rate and sensing Cram\'er--Rao {lower} bound (CRLB). {It is shown that under the equal power allocation strategy, the CRLB remains independent of the HRIS coefficients.}
Then, we formulate a joint optimization problem of power allocation and HRIS beamforming to maximize the communication sum-rate while ensuring specified sensing CRLB constraints. 
To solve the formulated non-convex problem, we propose an  alternating optimization algorithm based on fractional programming and successive convex approximation.
Extensive simulations validate our analysis and proposed algorithm, showing significant improvements in both communication and sensing performances enabled by the HRIS. 
{For example, an HRIS with only $4$ active elements offers $97.30\%$ improvement in the communication sum-rate, while ensuring a sensing CRLB constraint of $-30$ dB.}

\end{abstract}

\section{{Introduction}} 
\IEEEPARstart{S}{ixth}-generation (6G) wireless technologies are expected to fulfill highly demanding objectives to support a wide range of advanced applications. These objectives include sustaining high throughput and ultra-low latency communication under stringent energy-efficiency constraints \cite{f_liu_isac_designs}. In addition, 6G aims to provide high-accuracy sensing capabilities through integrated sensing and communication (ISAC),
\cite{liu_beampattern}, \cite{buzzi1}. 
However, the power-sharing trade-off between communication and sensing functionalities remains a key challenge in ISAC systems. In particular,
 strict constraints on the sensing performance  can severely limit the power available for communication.
In this context, hybrid reconfigurable intelligent surface (HRIS) has emerged as a promising solution, due to its ability to offer both passive reflection and active amplification gains to significantly enhance the power of the received signal \cite{nhan_SE}.
 The existing literature \cite{nhan_hris_tvt}, \cite{uniyal_spawc} has demonstrated that HRIS can achieve considerable performance gains when the transmit power is limited, making it an appropriate use case for scenarios where sensing requirements are stringent. 
Therefore, the integration of HRIS into massive multiple-input-multiple-output (mMIMO)-ISAC systems holds synergistic potential to improve both communication and sensing performances.
This work aims to realize these potential gains through a comprehensive analysis of the communication sum-rate and sensing Cramér-Rao lower bound (CRLB), along with joint optimization of the transceiver and HRIS beamforming for HRIS-assisted mMIMO
ISAC systems.
\subsection{Related Works}

A wide body of research has focused on transmit beamforming design for multiuser, multi-antenna ISAC systems, aiming to jointly achieve high communication rate and sensing accuracy \cite{peng_mi,liu_beampattern,nhan_JCAS,uniyal_wcnc,topal}. 
Early research on ISAC largely focused on conventional MIMO systems \cite{liu_beampattern}, \cite{liu_mimo}, however, the use of small-scale antenna arrays in these systems provided limited beamforming gains \cite{buzzi1}. 
In contrast, mMIMO leverages large antenna arrays to provide vast spatial degrees of freedom, which can improve both communication and sensing performances \cite{ngo}, \cite{buzzi}. 
Consequently, several works have focused on the design and optimization of mMIMO-based ISAC systems \cite{liao,nhan_isac_power_allocation,nguyen_jsac_energy_efficiency,wang_mMimo_radar,meng_mmimo}. Liao \textit{et al.} \cite{liao} proposed a power allocation algorithm to minimize the total transmit power under signal-to-interference-plus-noise ratio (SINR) and mainlobe-to-average-sidelobe
ratio constraints. Nguyen \textit{et al.} \cite{nhan_isac_power_allocation} derived closed-form expressions for the communication rate and sensing CRLB. They also proposed a power-allocation strategy to maximize the sum-rate under the CRLB constraints. {In \cite{zhao_mmwave}, \cite{zhang_tensor} efficient
channel and target parameter estimation approaches were developed for millimeter-wave mMIMO ISAC systems,
relying on compressed sampling \cite{zhao_mmwave} and tensor frameworks
\cite{zhang_tensor}.
Li \textit{et al.} \cite{bin_liao_icassp} proposed low-complexity transmit beamforming designs for maximizing the sensing energy, based on instantaneous channel state information (CSI) and statistical CSI, by using a 
dual-functional beamformer.} 

Despite these numerous advances, the fundamental power-sharing trade-off in ISAC has motivated researchers to explore the integration of RIS into ISAC systems \cite{wang,r_liu_ris,lago_isac_ris,anton_isac_hris}.
{The performance benefits of RIS-assisted wireless systems have been analyzed
under a variety of fading models. Specifically, Rayleigh and Rician fading models
were investigated in \cite{LU2026URLLC} and \cite{nhan_SE},
respectively, as they enable tractable analysis for rich-scattering and
line-of-sight propagation environments. Furthermore, $\alpha$-$\mu$, $\kappa$-$\mu$, and $\eta$-$\mu$ fading models were considered in \cite{kong2021effective,da2024alpha}, while Alvarado \textit{et al.} \cite{alvarado2025performance} investigated a mixed fading model. Do \textit{et al.} \cite{do2021aerial} modeled small-scale fading using Nakagami-$m$ distribution and large-scale shadowing using an inverse-Gamma distribution. 
More recently, Le \textit{et al.} \cite{le2026analysis} considered the generalized $\alpha$-$\eta$-$\kappa$-$\mu$ fading model, where the performance gains offered
by RIS assistance were quantified in terms of the block-error rate and
ergodic rate. Building on these demonstrated RIS gains under diverse propagation conditions, RIS-aided ISAC designs have recently been investigated to improve the communications--sensing tradeoff.}
{In\cite{r_liu_ris}, \cite{xing_isac_ris}, the transmit beamforming and RIS reflection coefficients were jointly optimized to enhance the communication performance under specified constraints on the sensing performance. 
Conversely, 
Song \textit{et al.} \cite{x.song} derived the CRLB for both point-target and extended-target sensing models and jointly designed the base station (BS) transmit beamforming and RIS reflective beamforming to minimize the resulting CRLB.
Chen \textit{et al.} \cite{chen_ris} 
proposed two waveform designs to enhance the communication rate--sensing beampattern error tradeoff and the energy efficiency.}

{It is worth noting that in the aforementioned RIS-aided ISAC systems, the 
RIS is a nearly-passive device with no power amplification capability. 
In this context, the works in \cite{wang,shao_ris,song_semi_passive} considered incorporating sensing capabilities at the RIS. However, this approach requires a more complex hardware design at the RIS and poses challenges for integration with the existing communication protocols \cite{yu_active_ris_isac}.
On the other hand, the works in \cite{yu_active_ris_isac,zhu_isac_ris,kumar_isac_aris,zhang_activeris} considered fully active RIS (ARIS)-assisted ISAC systems. Specifically, Yu \textit{et al.} \cite{yu_active_ris_isac} proposed a majorization-minimization based algorithm for joint transmit beamforming and active RIS design and provided analytical insights into the scaling laws of the radar SINR.
Zhang \textit{et al.} \cite{zhang_activeris} considered a cloud radio access network scenario and optimized the radar beampattern towards the sensing targets.}  
Although the performance gains
offered by an active RIS are well investigated in communication 
systems, studies confirm that excessive active elements can
degrade the performance under a limited power budget \cite{nhan_SE,nhan_hris_tvt}. An HRIS, on the other hand, provides an effective balance
to this tradeoff, enabling both reflection and amplification of
the signals with a few active elements \cite{nhan_hris_tvt}. Recently, its integration into ISAC system was investigated in  \cite{liao_hris_isac,shankar_hris_isac,zehra_hris_isac}.
Specifically, Liao \textit{et al.} \cite{liao_hris_isac} focused on maximizing the worst-case sensing beampattern gain under per-user SINR constraints, using an alternating-optimization method, while accounting for target location uncertainty.
Also, Yigit \textit{et al.} \cite{zehra_hris_isac}, considered a hybrid simultaneous transmission and reflection RIS and focused on maximizing the communication SINR under a {CRLB} constraint, using semidefinite relaxation. {In \cite{yao_hris} and \cite{saikia_hris}, deep reinforcement learning based frameworks were proposed to optimize the sum secrecy rate and the sum-rate, respectively.
Lin \textit{et al.} \cite{lin_hris} proposed a joint mode-selection and beamforming design aimed at maximizing the minimum sensing beampattern gain across multiple targets.}

\subsection{ Contributions}
The existing works on RIS-aided ISAC systems typically assume that the BS and RIS are equipped with uniform linear arrays (ULAs) \cite{ x.song,yu_active_ris_isac,zhu_isac_ris}. Consequently, their optimization frameworks are limited to estimating only a single direction-of-arrival parameter, such as the azimuth angle. Moreover, the optimization and design in these works is typically based on fast time-varying small-scale fading channel coefficients, making them computationally expensive. 
Furthermore, the four-hop echo signal path, i.e., BS-RIS-target-RIS-BS, assumed in most prior works may result in a severely attenuated echo signal.

In this work, we consider an HRIS-assisted monostatic mMIMO ISAC system for point-target sensing. Employing uniform planar arrays (UPAs) at the BS and the HRIS, we characterize the CRLB for the target's azimuth and elevation angles. 
For the communication subsystem, we analyze the achievable rate of the users, with linear precoders, including zero-forcing (ZF) and maximum ratio transmission (MRT). Our joint design problem of HRIS beamforming and power allocation is formulated relying on the knowledge of large-scale fading coefficients leading to practical solutions. {The existing analysis and designs cannot be directly applied to
our considered system. This is because we employ a UPA at the BS and
formulate the power-allocation problem based on large-scale fading
coefficients.} 
{Furthermore, unlike most prior works, wherein pilot contamination was ignored \cite{liao,wang,r_liu_ris,x.song}, we account for pilot contamination and analyze its impact on the communication and sensing  performances.}

Our key contributions are summarized as follows:
\begin{itemize}
    \item We first model the uplink and downlink signals of the HRIS-aided mMIMO ISAC system and derive the minimum-mean-square-error (MMSE) estimates of the effective channels under time-division duplex (TDD) operation. These channel estimates enable the design of a linear dual-function transmit beamformer. This result overcomes the limitations of prior works that rely on the assumption of perfect CSI knowledge \cite{liao_hris_isac,shankar_hris_isac,zehra_hris_isac}. 
\item We then derive closed
form expressions for ISAC performance metrics, including the
achievable communication rate  and the
CRLB for both the azimuth and elevation angles associated
with the sensing target. These expressions are unified for both MRT and ZF precoding and are given as functions of the large-scale fading parameters. Our
derivations are also valid for conventional 
mMIMO ISAC systems with and without passive RISs,
as these represent special cases of the considered HRIS-aided system.
 Furthermore, we investigate the sensing-communication performance tradeoff with
different RIS architectures and linear beamformers. To the best of our knowledge, this investigation has not been provided in the existing works.
\item {Our analysis reveals that the normalized mean square error (NMSE) of the MMSE channel estimate follows an inverse power-scaling law with respect to the total number of HRIS elements $N$. Consequently, the required pilot transmit power can be reduced proportionally as $N$ increases. Moreover, under the equal power allocation, 
the communication sum-rate exhibits the same inverse power-scaling law with respect to $N$, whereas the sensing CRLB remains independent of the HRIS coefficients.}

\item Based on the derived analytical expressions, we formulate an optimization problem to maximize the sum-rate, subject to constraints on the HRIS coefficients and transmit power at the BS. We solve the formulated  non-convex problem by proposing an alternating optimization (AO) algorithm based on fractional programming (FP) and successive convex approximation (SCA).
\item Finally, we perform extensive simulations to verify the analytical expressions and effectiveness of the proposed algorithm.
Simulation results demonstrate the performance gains of the HRIS-aided mMIMO ISAC system over mMIMO ISAC systems with and without a passive RIS (PRIS).
\end{itemize}     

{Unlike most existing HRIS-assisted ISAC works \cite{liao_hris_isac,shankar_hris_isac,zehra_hris_isac}, which assume perfect CSI and small-scale systems, we consider MMSE channel estimation and derive closed-form expressions for the communications sum rate and the CRLBs based solely on large-scale system parameters. Although closed-form expressions for the achievable rate and CRLBs have
been derived for conventional massive MIMO-ISAC systems in
\cite{nhan_isac_power_allocation}, these results are not
valid in HRIS-assisted massive MIMO-ISAC systems because
the deployment of the HRIS fundamentally changes the effective channels.}

\subsubsection*{Organization}
The rest of this paper is organized as
follows. We introduce the system model and
MMSE estimation of the effective
uplink channels in Section~\ref{sec:channel}. Section~\ref{sec:signal_model} presents the ISAC signal models and transmit beamformer model.  
In Section~\ref{sec:performance_analysis}, we derive closed-form expressions for ISAC performance metrics. Section~\ref{sec:problem_formulation} details the sum-rate maximization problem and its solution. Simulation
results are provided in Section~\ref{sec:simulation_results},
while Section~\ref{sec:conclusion} concludes the paper.

\subsubsection*{Notation}
Vectors and matrices are represented by
lowercase and uppercase bold letters, respectively. The expectation of a random variable is denoted by
$\mathbb{E}\{\cdot\}$, while $\mathbb{C}\{\cdot\}$ and $\mathbb{C}\{\cdot,\cdot\}$ are the auto- and cross-covariance operators, respectively; $\mathtt{tr}(\cdot)$ and $\mathtt{vec}(\cdot)$, respectively denote the trace and vectorization operators, while $(\cdot)^\ast$, $(\cdot)^{\mathsf T}$ and $(\cdot)^{\mathsf H}$ are the conjugate, transpose and conjugate-transpose operators, respectively. We denote by $\dot{f_x}$ the partial derivative of $f$ with respect to $x$, i.e.,
$\dot{f_x} \triangleq {\partial f}/{\partial x}$, while $\circ$ denotes a Hadamard product. The magnitude of a complex number and the Euclidean norm of a vector are denoted by $|\cdot|$ and  $\|\cdot\|$, respectively.
The complex Gaussian distribution with zero mean and variance $\sigma^{2}$ is denoted by $\mathcal{CN}(0,\sigma^{2})$. 


\begin{figure}[t]
	\centering
	\includegraphics[scale=0.2]{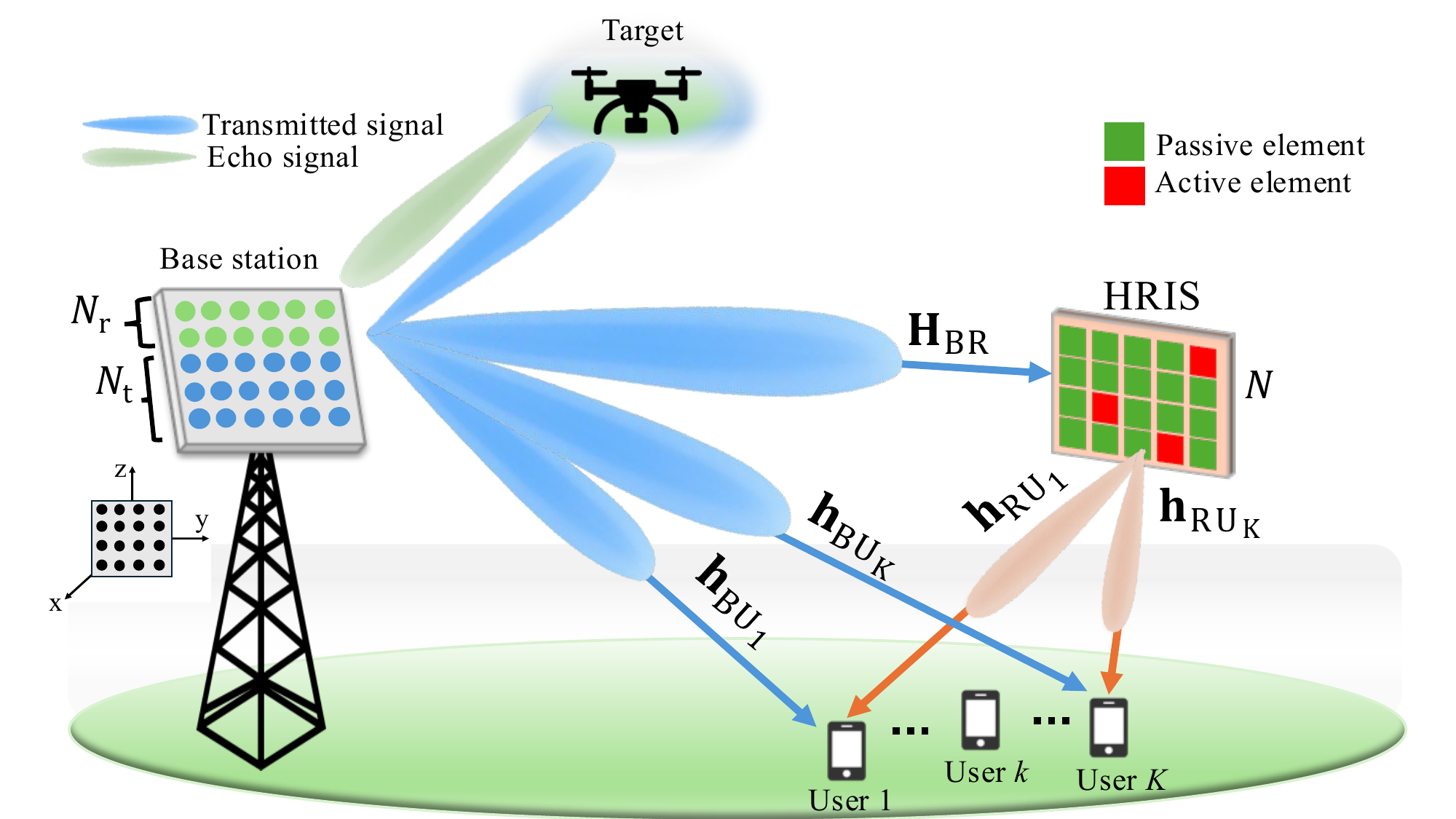}  
	\caption{\small HRIS-aided mMIMO ISAC system with a UPA-equipped BS.}
	\label{fig:schematic}
\end{figure}

\section{System  Model and Channel Estimation}
\label{sec:channel}
We consider a downlink HRIS-assisted monostatic mMIMO ISAC system, as illustrated in Fig.~\ref{fig:schematic}. The BS is equipped with a UPA comprising
$N_\mathtt{t}$ transmit antennas and $N_\mathtt{r}$ radar receive antennas.
The BS simultaneously transmits data signals to $K$ single-antenna
downlink users while also transmitting radar probing signals toward the sensing target. {Specifically, we focus on the tracking phase, where the target has been previously detected during a search phase} \cite{nguyen_jsac_energy_efficiency}.
{Furthermore, the communication between the BS and the users is assisted by the HRIS. The BS--target link is assumed to be established only via the direct path, which overcomes the four-hop path loss incurred by RIS-assisted sensing considered in \cite{x.song,r_liu_ris,wang}. Moreover, HRIS-aided communication can achieve a required sum-rate with lower transmit power \cite{nhan_hris_tvt,uniyal_spawc}, which can enable more power to be allocated for sensing.
}

{The HRIS is equipped with a UPA of $N$ elements,} including $N_\mathsf{a}$ active amplifying elements and $N-N_\mathsf{a}$ passive reflecting elements. To guarantee a minimal increase in power consumption and hardware cost of
the HRIS, only a few active elements are employed, i.e., $N_\mathsf{a}\ll N$ \cite{nhan_SE},\cite{nhan_hris_tvt}.
We denote by $\mathcal{I}_\mathsf{a}\subset \{1,2,\dots,N\}$, the predefined index set of active elements in the HRIS with $|\mathcal{I}_\mathsf{a}|=N_\mathsf{a}$.
Let $\alpha_n$ be the coefficient associated with the $n$-th element of the HRIS. It can be expressed as $\alpha_n=|\alpha_n|e^{j\vartheta_n}$, where $\vartheta_n$ and $|\alpha_n|$ are the phase and amplitude of the $n$-th coefficient of the HRIS, respectively. In particular, $\vartheta \in [0,2\pi)$, $\forall n$,  while $|\alpha_n|=1$ for $n \notin \mathcal{I}_\mathsf{a}$ and
$|\alpha_n|\leq \alpha_{\max}$ for $n \in \mathcal{I}_\mathtt{a}$. Here, $\alpha_{\max}$
is the maximum gain that an active element can provide. This implies that, unlike passive elements, the amplitudes of the active coefficients can be optimized to enhance the system performance. Denote by  $\boldsymbol{\Theta}_\mathtt{R}\triangleq \text{diag}\{\alpha_1,..., \alpha_{{N}}  \}$ the coefficient matrix of the HRIS.
 We introduce an additive decomposition for it, given as
$\boldsymbol{\Theta}_\mathtt{R} =\boldsymbol{\Theta}_\mathtt{p}+ \boldsymbol{\Theta}_\mathtt{a}$, 
where 
$
\boldsymbol{\Theta}_\mathtt{a}= \mathbbm{1} \circ \boldsymbol{\Theta}_\mathtt{R}$ and $\boldsymbol{\Theta}_\mathtt{p}= (\mathbf{I}_N - \mathbbm{1}) \circ \boldsymbol{\Theta}_\mathtt{R}$ are both diagonal matrices 
containing only active and passive coefficients
of the HRIS, respectively.
Here, $\mathbbm{1}$ represents an \( N \times N \) diagonal matrix, whose non-zero diagonal elements are unity, with positions determined by the set $\mathcal{I}_\mathsf{a}$.

\subsection{Channel Model and Uplink Training}
\label{subsection:channel_model}

\subsubsection{Channel Model}
Let $\mathbf{h}_{\mathtt{UR}_k} \in \mathbb{C}^{N \times 1}$, $\mathbf{H}_{\mathtt{RB}} \in \mathbb{C}^{N_\mathtt{t}\times N}$, and $\mathbf{h}_{\mathtt{UB}_k} \in \mathbb{C}^{N_\mathtt{t} \times 1}$ denote the uplink channels from the $k$-th user to the HRIS, from the HRIS to the BS, and from the $k$-th user to the BS, respectively. In particular, $\mathbf{h}_{\mathtt{UR}_k}$, $\mathbf{h}_{\mathtt{UB}_k}$, and the $n$-th column of $\mathbf{H}_{\mathtt{RB}}$ can be written as
 \begin{align}
\mathbf{h}_\mathtt{X}={\beta^{\frac{1}{2}}_\mathtt{X}}\Tilde{\mathbf{h}}_\mathtt{X}, \quad \mathtt{X}\in\{\mathtt{UR}_k,\mathtt{UB}_k,\mathtt{RB}\},     
 \end{align}
where 
${\beta_{\mathtt{X}}}$ and $\Tilde{\mathbf{h}}_{\mathtt{X}}$ are the large-scale fading coefficient and  small-scale fading channel vector, respectively, and $\Tilde{\mathbf{h}}_{\mathtt{X}}$ follows a Rayleigh fading model. 
The effective channel between the $k$-th user and the BS, denoted by $\mathbf{g}_k \in \mathbb{C}^{N_\mathtt{t}\times 1}$, is given as
\begin{align} \label{channel}\mathbf{g}_k&=\mathbf{h}_{\mathtt{UB}_k}+\mathbf{H}_{\mathtt{RB}}\boldsymbol{\Theta}_\mathtt{R}\mathbf{h}_{\mathtt{UR}_k} \nonumber \\
&=\mathbf{h}_{\mathtt{UB}_k}+\sum\nolimits_{n=1}^N\alpha_n\mathbf{h}_{\mathtt{RB}_n}h_{\mathtt{UR}_{k,n}}.
\end{align}

\subsubsection{Uplink Channel Estimation}
\label{section:channel_estimation}
{We consider a TDD protocol, in which uplink pilot training is first performed for channel estimation, followed by downlink communication and sensing transmission. TDD is adopted to exploit uplink--downlink channel reciprocity. In particular, the channel-acquisition overhead required for downlink precoding is \(K\) pilot dimensions in TDD, whereas it is $\frac{1}{2}\left( N_{\mathrm{t}}+K+\max(N_{\mathrm{t}},K) \right)$ in frequency-division duplex (FDD) systems~\cite[Section~1.3.5]{bjornson2017massive}. Consequently, the FDD overhead becomes substantially larger in the mMIMO regime \(N_{\mathrm{t}}\gg K\), making TDD more efficient for the considered system.}
 In the uplink training phase, all the communication users send pilot sequences to the BS for channel estimation. We denote by $\boldsymbol{\psi}_k \in \mathbb{C}^{\tau_\mathtt{p} \times 1}$ the pilot sequence sent by user $k$. Here, $\tau_\mathtt{p}$ is the length of the pilot sequence and $\tau_\mathtt{p} < \tau_\mathtt{c}$, where $\tau_\mathtt{c}$ is the number of samples in one coherence interval. 
Let $\mathcal{P}_k$ be the set of user indices (including user $k$) that
reuse the pilot sequence assigned to user $k$. The pilot sequences are considered to be mutually orthogonal, satisfying $\boldsymbol{\psi}_j^{\H} \boldsymbol{\psi}_k = 1$ if $j \in \mathcal{P}_k$, and $\boldsymbol{\psi}_j ^{\H}\boldsymbol{\psi}_k = 0$ otherwise.
 The received pilot signal 
 at the BS, denoted by $\mathbf{Y}_\mathtt{p}\in \mathbb{C}^{N_\mathtt{t} \times \tau_\mathtt{p}}$, is given as
\begin{align}
\mathbf{Y}_\mathtt{p}=\sqrt{\tau_\mathtt{p}p_\mathtt{p}}\sum\nolimits_{k=1}^K\mathbf{g}_k\boldsymbol{\psi}^\H_k +
\mathbf{N}_\mathtt{p}+ \mathbf{H}_{\mathtt{RB}}\boldsymbol{\Theta}_\mathsf{a} \mathbf{N}_\mathsf{a}  ,   
\end{align}  
 where $p_\mathtt{p}$ is the power of each user to transmit the pilot signal, and $\mathbf{N}_\mathtt{p} \in \mathbb{C}^{N_\mathtt{t} \times \tau_\mathtt{p}}$ is the additive white Gaussian noise (AWGN) matrix with entries distributed as $\mathcal{CN}(0,\sigma^2_\mathtt{p})$; $\mathbf{N}_\mathsf{a}\in \mathbb{C}^{N\times \tau_\mathtt{p}}$ is the noise matrix due to the active elements of the HRIS, whose entries are $\mathcal{CN}(0,\sigma^2_\mathsf{a})$ random variables.
Accordingly, $\mathbf{Y}_\mathtt{p}$ can be rewritten as 
\begin{align}
\mathbf{Y}_\mathtt{p}=\sqrt{\tau_\mathtt{p}p_\mathtt{p}}\sum\nolimits_{k=1}^K\mathbf{g}_k\boldsymbol{\psi}^\H_k +\mathbf{N},   
\end{align} 
where
\begin{align}
 \mathbf{N}  \triangleq   \mathbf{H}_{\mathtt{RB}}\boldsymbol{\Theta}_\mathsf{a} \mathbf{N}_\mathsf{a} +\mathbf{N}_\mathtt{p} =  \sum_{n=1}^N \alpha_n \mathbf{h}_{\mathtt{RB}_n} \mathbf{n}_{\mathsf{a},n}^\H+\mathbf{N}_\mathtt{p},
\label{uplink_noise}
\end{align}
is the aggregated noise matrix at the BS. Here, $\mathbf{n}_{\mathsf{a},n}$ is the $n$-th row of $\mathbf{N}_\mathsf{a}$. From \eqref{uplink_noise}, we note that $\mathbf{H}_\mathtt{RB}$, $\mathbf{N}_\mathsf{a}$, and $\mathbf{N}_\mathtt{p}$ are mutually independent. Furthermore, we recall that $\boldsymbol{\Theta}_\mathsf{a}$ is a diagonal matrix, whose non-zero diagonal entries are determined by the positions of the active elements in
$\mathcal{I}_\mathsf{a}$. As a result, the entries of $ \mathbf{N}$ have zero-mean and variance $\sigma_\mathtt{z}^2$, given by
\begin{align}
{\sigma_\mathtt{z}^2(\check{\bm{\alpha}})}=   \sigma_\mathtt{p}^2 + \beta_\mathtt{RB} \sum_{n \in \mathcal{I}_\mathsf{a}} |\alpha_n|^2 \sigma_{\mathsf{a}}^2,\label{aggregate_noise}
\end{align}
{where $\check{\bm{\alpha}}=\left[\check{\alpha}_1,\cdots, \check{\alpha}_N\right]$, with $\check{\alpha}_n=|\alpha_n|^2$.}
{Note that the second term in \eqref{aggregate_noise} does not appear in mMIMO systems without an RIS or with a PRIS. Thus, the existing results in the literature do not directly apply to our work. In the following, we derive the MMSE estimate of the effective channel $\mathbf{g}_k$ for the considered HRIS-aided mMIMO ISAC system.}
{To estimate $\mathbf{g}_k$, the BS multiplies $\mathbf{Y}$ with
$\boldsymbol{\psi}_k$ to obtain}
\begin{align}
\tilde{\mathbf{y}}_k\triangleq \mathbf{Y}\boldsymbol{\psi}_k=\sqrt{\tau_\mathtt{p}p_\mathtt{p}}\mathbf{g}_k+  \sum\nolimits _{j \in \mathcal {P}_{k} \backslash \{k\}} \sqrt{\tau_\mathtt{p}p_\mathtt{p}}\mathbf{g}_{j}+\mathbf{n}_k,
 \label{estimate}
\end{align}
where   
$\mathbf{n}_k=\mathbf{N}\boldsymbol{\psi}_k$.
From \eqref{estimate}, the corresponding estimate of $\mathbf{g}_k$ is given in the following theorem.
\begin{theorem} \label{thm:mmse_estimate}
  The MMSE estimate of $\mathbf{g}_k$ is given by
\begin{align} 
    \hat { {\mathbf{g}}}_{k} &= \sqrt{\tau_\mathtt{p} p_\mathtt{p}}\mathbf{R}_k \boldsymbol{\Psi}_k^{-1}  \tilde{\mathbf{y}}_{k} ,
\label{eq:ghat}\end{align}
where 
\begin{align}
\mathbf{R}_k&\triangleq \underbrace{\left(\beta_{\mathtt{UB}_k} +\sum _{n=1}^{N}
|\alpha_n|^2
\beta_{\mathtt{UR}_k}\beta_{\mathtt{RB}}\right)}_{\triangleq \varrho_k}\bold{I}_{N_\mathtt{t}},
\label{eq:cov_gk}
\\
\boldsymbol{\Psi}_k&\triangleq \left(\sum _{k \in \mathcal {P}_{k} } \tau_\mathtt{p} p_\mathtt{p} {\varrho}_k+\sigma_\mathtt{p}^2 + \beta_\mathtt{RB} \sum_{n \in \mathcal{A}} |\alpha_n|^2 \sigma_{\mathsf{a}}^2\right)\bold{I}_{N_\mathtt{t}}.
 \label{psi_k}
\end{align}
The channel estimate $\hat { {\mathbf{g}}}_{k} $ and channel estimation error $\Tilde{\bold{g}}_k=\bold{g_k}-\hat{\bold{g}}_k$ are independently distributed as
\begin{align}
   \hat { {\mathbf{g}}}_{k}\sim \mathcal{CN}\left(\mathbf{0},\hat{{\varrho}}_k\mathbf{I}_{N_\mathtt{t}}\right) \,\,\text{and} \,\,  \tilde{ {\mathbf{g}}}_{k}\sim \mathcal{CN}\left(\mathbf{0},\tilde{{\varrho}}_k\mathbf{I}_{N_\mathtt{t}}\right),
   \label{channel_estimate_error}
\end{align}
respectively, where
\begin{align}
    \label{rk_hat}\hat{{\varrho}}_k\triangleq \frac{\tau_\mathtt{p} p_\mathtt{p} \varrho_k^2}{\sum _{k \in \mathcal {P}_{k} } \tau_\mathtt{p} p_\mathtt{p} {\varrho}_k+\left(\sigma_\mathtt{p}^2 + \zeta_\mathtt{RB} \sum_{n \in \mathcal{A}} |\alpha_n|^2 \sigma_{\mathsf{a}}^2\right)}
\end{align}
and $\Tilde{{\varrho}}_k\triangleq{\varrho}_k-\hat{{\varrho}}_k$.
\end{theorem}


{The proof follows similar steps to those in \cite[Appendix~A]{uniyal_wcnc}. Unlike \cite{uniyal_wcnc}, which considers a fully passive RIS, we study a hybrid active-passive RIS, which leads to a different effective channel model. After substituting the effective channel in \eqref{channel} into the received signal in \eqref{estimate}, the remaining steps are analogous and are omitted due to space limitations.}

The NMSE of the channel estimate of the $k$-th user in \eqref{eq:ghat} is given as
\begin{align}
\bar{\epsilon}_k &= \frac{\mathbb{E}\left[\|\tilde{\mathbf{g}}_k\|^2\right]}
       {\mathbb{E}\left[\|\mathbf{g}_k\|^2\right]}=1-\frac{\tau_\mathtt{p} p_\mathtt{p} \varrho_k}{\sum _{k \in \mathcal {P}_{k} } \tau_\mathtt{p} p_\mathtt{p} {\varrho}_k+\sigma_\mathtt{z}^2},
\label{eq:nmse}
\end{align}
where 
$\sigma_\mathtt{z}^2$ is given in \eqref{aggregate_noise}.
\begin{remark}
\label{rem:nmse_scaling}
It is observed from \eqref{eq:nmse} that as $p_\mathtt{p} \to \infty$, $\bar{\epsilon}_k \to 1-\frac{\varrho_k}{\sum_{k\in\mathcal{P}_k}\varrho_k}$, which is due to the pilot contamination. 
Furthermore, for large $N$ and finite $p_\mathtt{p}$, the NMSE in \eqref{eq:nmse} is approximated as 
\begin{align}
   \bar{\epsilon}_k\approx1-\!\left(\!{1\!+\!\frac{\sum_{j \in \mathcal{P}_k \backslash \{k\} }\beta_{\mathtt{UR}_j}} {\beta_{\mathtt{UR}_k}}\!+\!\frac{\sigma_\mathtt{z}^2}{\tau_\mathtt{p} p_\mathtt{p}N\beta_{\mathtt{UR}_k}\beta_{\mathtt{RB}}}}\right)^{-1}.
\label{eq:nmse_large_n}
\end{align}
Interestingly, \eqref{eq:nmse_large_n} shows that $\bar{\epsilon}_k$ also converges to a non-zero constant as $N\to \infty $, 
 demonstrating that increasing $N$ has an effect similar to increasing $p_\mathtt{p}$.
Notably, in the absence of pilot contamination, we have $\mathcal{P}_k \equiv \{k\}$ and thus $\bar{\epsilon}_k\to 0$ as $N\to \infty$.
Furthermore, to analyze the power-scaling laws, we assume that the pilot power is scaled as $p_\mathtt{p}=\frac{e_\mathtt{p}}{N^\varepsilon}$, where $e_\mathtt{p}$ is a constant and $\varepsilon>0$ is a scaling exponent. Specifically, setting
$p_\mathtt{p}=\frac{e_\mathtt{p}}{N^\varepsilon}$ in \eqref{eq:nmse_large_n} yields $\bar{\epsilon}_k < 1$ for $\varepsilon\leq1$, showing that an HRIS with large
 $N$ and few active elements can maintain the NMSE strictly below one, even at low pilot powers.
\end{remark}


\subsection{Communication and Radar Signal Model}
\label{sec:signal_model}

\subsubsection{Communication Subsystem}
We denote by $s_{k\ell}$ the symbol transmitted by the BS to the $k$-th user at the $\ell$-th time slot, for $\ell=1,\dots,L$, where $L$ is the length of the communication/radar frame. 
The symbol matrix for all users over $L$ time slots is given as $\mathbf{S}=\left[\mathbf{s}_1,\dots,\mathbf{s}_L\right]\in \mathbb{C}^{K \times L}$, where $\mathbf{s}_\ell=[s_{1\ell},\dots,s_{K\ell}] \in \mathbb{C}^{K \times 1}$ and $\mathbb{E}\left\{\mathbf{s}_\ell\mathbf{s}_\ell^\H\right\}=\mathbf{I}_K$. By the law of large numbers, we have ${\mathbf{S}\mathbf{S}^\H}/{L}\approx \mathbf{I}_K$ for sufficiently large $L$.
Let $
\mathbf{W}=[\mathbf{w}_1,\cdots,\mathbf{w}_K]\in \mathbb{C}^{N_\mathtt{t}\times K}$ be the linear dual-functional precoder employed by the BS.
Then, the
transmit signal matrix  is given as $\bold{X}=\mathbf{W}\bold{S} \in \mathbb{C}^{N_\mathtt{t}\times L}$. Specifically, the $\ell$-th column of $\mathbf{X}$, denoted as $\mathbf{x}_\ell$ is the transmit signal vector at the $\ell$-th time slot, and is given as $\mathbf{x}_\ell=\mathbf{W}\mathbf{s}_\ell=\sum_{k=1}^K\mathbf{w}_ks_{k\ell}$.

Similar to the uplink, the downlink channel comprises a 
direct link and an HRIS-assisted link, as shown in Fig.~\ref{fig:schematic}. By exploiting the channel reciprocity, we model the channels from the BS to the HRIS, from HRIS to user $k$, and from the BS to user $k$ as $\mathbf{H}_{\mathtt{RB}}^\H $,  $\mathbf{h}_{\mathtt{UR}_k}^\H$, and $\mathbf{h}_{\mathtt{UB}_k}^\H$, respectively.
Thus, the signal received by the $k$-th user at the $\ell$-th time slot is given as
\begin{align}
\label{rx_comms}
 \hspace{-0.1cm}\!y_{k\ell}\!=\!\underbrace{\left(\mathbf{h}^\H_{\mathtt{UB}_k}\!+\!\mathbf{h}^\H_{\mathtt{UR}_k}\boldsymbol{\Theta}_\mathtt{R}\mathbf{H}^\H_{\mathtt{RB}}\right) }_{\triangleq\mathbf{g}_k^\H}\sum_{k=1}^{K} \!\!\mathbf{w}_ks_{k\ell} + \underbrace{ {\mathbf{h}^\H_{\mathtt{UR}_k}\mathbf{\Theta}_\mathsf{a}\mathbf{n}_\mathsf{a} \!+\! n_{k\ell}}}_{\triangleq{n_{\mathtt{d},{k\ell}}}},
\end{align}
where
 $n_{k\ell} \sim \mathcal{CN}(0,\sigma^2_\mathtt{c})$ is the AWGN and $\mathbf{n}_\mathsf{a} \sim \mathcal{CN}(0,\sigma_\mathtt{a}^2 \mathbbm{1})$ is the noise caused by the active elements of the HRIS. Furthermore, ${{n_{\mathtt{d},{k\ell}}}}$ is the aggregated noise at the $k$-th user, modeled similar 
 to
 \eqref{aggregate_noise}. In particular, ${{n_{\mathtt{d},{k\ell}}}}$ has zero-mean and variance ${\sigma_{\mathtt{d},{k}}^2(\check{\bm{\alpha}})= \sigma^2_\mathtt{c} + \beta_{\mathtt{RU}_k}\sum_{n \in \mathcal{I}_\mathsf{a}}|\alpha_n|^2 \sigma_\mathsf{a}^2}$.

\subsubsection{Radar Subsystem}

For the sensing subsystem, we consider that the target lies in the line-of-sight of the BS. Without loss of generality, we assume that the UPA at the BS is placed on the $\mathtt{y}-\mathtt{z}$ plane and centered at the origin. The target is located in the direction $(\theta, \phi)$, where $\theta \in \left[-\pi,\pi\right]$ is the azimuth and $\phi \in \left[-\pi/2, \pi/2\right]$ is the elevation angle, relative to the $\mathtt{x}$-axis and $\mathtt{z}$-axis, respectively, as shown in Fig.~\ref{fig:schematic}. The transmit steering vector $\mathbf{a}(\theta, \phi)$ for this configuration is modeled as 
$\mathbf{a}(\theta, \phi)=\mathbf{a}_\mathtt{y}(\theta, \phi) \otimes \mathbf{a}_\mathtt{z}(\phi)$, where 
\begin{align}
&\mathbf{a}_\mathtt{y}=
\left[ e^{-j\pi \frac{N_\mathtt{ty} - 1}{2}\kappa_\mathtt{y}}, e^{-j\pi \frac{N_\mathtt{ty} - 3}{2}\kappa_\mathtt{y}
}, \right.\left. \dots,e^{j\pi \frac{N_\mathtt{ty} - 1}{2}\kappa_\mathtt{y} 
} 
 \right]^\T, \nonumber
 \\
 &\mathbf{a}_\mathtt{z}=
\left[ e^{-j\pi \frac{N_\mathtt{tz} - 1}{2} \kappa_\mathtt{z}}, e^{-j\pi \frac{N_\mathtt{tz}- 3}{2} \kappa_\mathtt{z}}, \right. \left. \dots,e^{j\pi \frac{N_\mathtt{tz} - 1}{2} \kappa_\mathtt{z}}
 \right]^\T.
 \label{steer_vector}
\end{align}
Here, $\kappa_\mathtt
y= \sin(\theta) \sin(\phi)$ and $\kappa_\mathtt
z= \cos(\phi)$ and
$N_\mathtt{ty}$ and $N_\mathtt{tz}$ are the number of transmit antennas along the $\mathtt{y}$ axis and $\mathtt{z}$ axis of the UPA, respectively, with $N_\mathtt{t}=N_\mathtt{ty}\times N_\mathtt{tz}$. 
The receive steering vector $\mathbf{b}(\theta, \phi)$ for $N_\mathtt{r}$ receive antennas at the BS can be
modeled similarly.  
Based on this, the round-trip channel between the BS and the target in the direction $(\theta,\phi)$ is given by
$ \mathbf{A}(\theta,\phi)=\mathbf{b}(\theta,\phi)\mathbf{a}^{\H}(\theta,\phi)$.  
Then, the echo signal received at the BS from the target can be expressed as
\begin{align}
\mathbf{Y}_{\mathtt{s}}=\beta_\mathtt{s}\mathbf{A}(\theta,\phi)\mathbf{X}+\mathbf{N}_\mathtt{s}, 
 \label{echo}
\end{align}   
where $\beta_\mathtt{s}$ is the complex reflection coefficient related to the {radar cross-section} of the target and $\mathbf{N}_\mathtt{s} \in \mathbb{C}^{{N}_\mathtt{R} \times L}$ is the AWGN matrix with  entries distributed as $\mathcal{CN}(0,\sigma^2_\mathtt{s})$. 
In the subsequent analysis, we omit $(\theta,\phi)$ for the ease of exposition.

\subsection{Dual-functional Transmit Precoder}
 
We assume that the BS employs MRT or ZF precoders for the communication users. Thus, the communication precoding matrix, represented by $\mathbf{F}^\mathtt{bf}=[\mathbf{f}^\mathtt{bf}_1, \cdots, \mathbf{f}^\mathtt{bf}_K] \in \mathbb{C}^{N_\mathtt{t}\times K}$ with $\mathtt{bf} \in \{\mathtt{MRT}, \mathtt{ZF}\}$, is given as
\begin{align}
\label{comms_beamformer}
\mathbf{F}^\mathtt{bf}=
\begin{cases}
\hat{\mathbf{G}}\triangleq \mathbf{F}^\mathtt{MRT}, & \text{for}\,\,\mathtt{ bf = MRT},\\
\hat{\mathbf{G}}\left(\hat{\mathbf{G}}^{\H}\hat{\mathbf{G}}\right)^{-1} \triangleq \mathbf{F}^\mathtt{ZF}, & \text{for}\,\,\mathtt{ bf = ZF}.
\end{cases}
\end{align}
Here, $\bold{\hat{G}} \triangleq [\bold{\hat{g}}_1,\dots,\bold{\hat{g}}_K] \in\mathbb{C}^{N_\mathtt{t}\times K}$ and $\bold{\hat{g}}_k$ is the estimated channel of the $k$-th user given in \eqref{channel}. 

For radar sensing, we denote by $\mathbf{u} \in \mathbb{C}^{K \times 1}$ the beamforming vector for sensing.
Then, the linear dual-function transmit precoder can be written as 
$
\mathbf{W}^\mathtt{bf}=\bold{F}^\mathtt{bf}\boldsymbol{\Xi}+\bold{u}\boldsymbol{\rho}^\T$. Here,  $\boldsymbol{\Xi}=\mathtt{diag}[\sqrt{\xi_1},\dots,\sqrt{\xi_K}]\in \mathbb{C}^{K \times K}$,   $\boldsymbol{\rho}=[\sqrt{\rho_1},\dots,\sqrt{\rho_K}]\in \mathbb{C}^{K \times 1}$, and $\xi_k$ and $\rho_k$ are the power coefficients allocated to the $k$-th user and sensing target in each data stream. Based on this, the echo signal in  \eqref{echo}, can be rewritten as
\begin{align}
 \mathbf{Y}_{\mathtt{s}}&=\beta_\mathtt{s}\mathbf{b}\mathbf{a}^\H\mathbf{W^\mathtt{bf}S}+\mathbf{N}_\mathtt{s}\nonumber  \\
 &=\beta_\mathtt{s}\left(\mathbf{b}\mathbf{a}^\H\bold{F}_\mathtt{c}^\mathtt{bf}\boldsymbol{\Xi}+\mathbf{b}\mathbf{a}^\H\bold{u}\boldsymbol{\rho}^\T\right)+\mathbf{N}_\mathtt{s}.
 \label{echo_optimal}
\end{align} 
From \eqref{echo_optimal}, the optimal sensing precoder that maximizes the echo signal for an arbitrary $\bold{F}^\mathtt{bf}\boldsymbol{\Xi}$ and known target angles $(\theta,\phi)$ is $\mathbf{u}=\mathbf{a}(\theta,\phi)$. {To maintain generality, we consider the sensing beamformer
$\vu=\va(\hat{\theta},\hat{\phi})$, where
$\hat{\theta}=\theta+\delta$ and
$\hat{\phi}=\phi+\delta$. Here, $\delta$ denotes the target-angle error.  This allows us to evaluate the sensing performance for the case where the true target angles are unknown.}
As a result, $\mathbf{W}^\mathtt{bf}$ can be expressed as
\begin{align}
\mathbf{W}^\mathtt{bf} =
\begin{cases}
\hat{\mathbf{G}}\boldsymbol{\Xi}+\mathbf{u}\boldsymbol{\rho}^\T 
, & \text{for}\,\,\mathtt{ bf = MRT}, \\
 \hat{\mathbf{G}}\left(\hat{\mathbf{G}}^{\H}\hat{\mathbf{G}}\right)^{-1}\boldsymbol{\Xi}+{\mathbf{u}}\boldsymbol{\rho}^\T 
 , & \text{for}\,\,\mathtt{ bf = ZF},
 \label{dual_bf}
\end{cases}
\end{align}
where the $k$-th column of $\mathbf{W}^\mathtt{bf}$ is $\mathbf{w}^\mathtt{bf}_k=\sqrt{\xi_k}\mathbf{f}_{k}^\mathtt{bf}+\sqrt{\rho_k}{\mathbf{u}}$.

To facilitate the subsequent analysis and optimization, we compute the average transmit power of the BS using the dual-function precoder given in \eqref{dual_bf} in the following lemma. 
\begin{lemma}
\label{lemma:transmit_covariance}
The total average transmit power at the BS employing the MRT and ZF precoder is given as
\begin{align}
\label{eq:total_transmit_power}
P^\mathtt{bf}_\mathtt{tot}(\boldsymbol{\xi},\eta,\boldsymbol{\alpha})={ P_\mathtt{c}^\mathtt{bf}} + {P_\mathtt{s}}, 
\end{align}
where $\mathtt{bf}\in \{\mathtt{MRT},\mathtt{ZF}\}$, $P_\mathtt{c}^\mathtt{bf}\triangleq N_\mathtt{t}{\sum_{k=1}^K\xi_k \nu_k^{\mathtt{bf}}}$, $P_\mathtt{s}\triangleq{N_\mathtt{t}{\eta}}$, and 
\begin{align}
\label{nu}
\nu_k^\mathtt{bf}\triangleq 
\begin{cases}
\hat{\varrho}_k , & \text{for}\,\,\mathtt{ bf = MRT},\\
\dfrac{\hat{\varrho}_k^{-1}}{N_\mathtt{t}(N_\mathtt{t}-K)} , & \text{for}\,\,\mathtt{ bf = ZF},
\end{cases}
\, \forall k=1,\dots,K.
\end{align}

\end{lemma}

\begin{proof}
    See Appendix~\ref{sec:transmit_covaraince}.
\hfill 
\qed
\end{proof}


\section{Performance Metrics for Communication and Sensing
}
\label{sec:performance_analysis}
We now derive closed-form expressions for the achievable communication rate and sensing CRLB to evaluate the performances of the communication and sensing subsystems, respectively. Subsequently, these expressions are utilized to formulate the
sum-rate maximization problem. 

\subsection{Communication Achievable Rate}

{
We rewrite the received signal in \eqref{rx_comms} as 
\begin{align} 
y_{k}\left[\ell\right]&=\mathbb{E}\left\{\bold{g}^\H_k\mathbf{w}_k\right\}s_{k}[\ell]+ 
 \sum_{j\neq k}\bold{g}^\H_k\mathbf{w}_js_{j}[\ell]\nonumber \\
&\hspace{1cm} +\left(\bold{g}^\H_k\mathbf{w}_k-\mathbb{E}\left\{\bold{g}^\H_k\mathbf{w}_k\right\}\right)s_{k}[\ell]+n_{\mathtt{d},{k}}[\ell]. 
\label{rx_comms2}
 \end{align}
From \eqref{rx_comms2}, the achievable rate of the $k$-th user is given by 
\begin{align}
\mathcal{R}_k= \tau_\mathtt{0}\log_2\left(1+\varsigma_k\right) \quad \text{[bits/sec/Hz]},
\label{rate_defination}
\end{align}
where  $\tau_\mathtt{0}\triangleq\left({\tau_\mathtt{c}-\tau_\mathtt{p}}\right)/{\tau_\mathtt{c}}$ and 
\begin{align}\label{eq_SINR}
\varsigma_k=\frac { \left |{\mathbb{E}\left\{\bold{g}^\H_k\mathbf{w}_k\right\}}\right |^{2}}{ \mathbb{E}\{|\bold{g}^\H_k\mathbf{w}_k|^2\} -|\mathbb{E}\{\bold{g}^\H_k\mathbf{w}_k\}|^2+ \sum\limits_{j\neq k}^{K} \mathbb{E}\{|\bold{g}^\H_k\mathbf{w}_j|^2\} + \sigma_{\mathtt{z}}^2}.  
\end{align}
We present a closed-form expression for the achievable rate in the following theorem.

\begin{theorem} \label{thm:sum_rate}
In the considered HRIS-assisted mMIMO ISAC system, the achievable rate for the \( k \)-th user employing the MRT and ZF precoder is given by
\begin{equation} \label{rate_final}
\mathcal{R}_k^{\mathtt{bf}}(\bm{\xi},\eta,\check{\bm{\alpha}})=  \tau_\mathtt{0}\log_2\left(1+\varsigma_k^{\mathtt{bf}}\right),\end{equation}   
where
$\mathtt{bf}=\{\mathtt{MRT,ZF}\}$, $\check{\bm{\alpha}}=\left[\check{\alpha}_1,\cdots, \check{\alpha}_N\right]$, $\check{\alpha}_n=|\alpha_n|^2$, and 
\begin{align}
   \label{sinr_final}\varsigma_k^{\mathtt{bf}}=\frac { {u_k^{\mathtt{bf}}(\check{\bm{\alpha}})\xi_k}}{\eta d_k (\check{\bm{\alpha}}) + \sum _{j=1}^{K} \xi_j c_{kj}^{\mathtt{bf}} (\check{\bm{\alpha}})+ \sigma_{\mathtt{z}}^2(\check{\bm{\alpha}})}. 
\end{align}
Here,  
$\eta\triangleq  \|\boldsymbol{\rho}\|^2$,  $d_k(\check{\boldsymbol\alpha})\triangleq N_\mathtt{t}\varrho_k$, and $u_k^\mathtt{bf}(\check{\bm{\alpha}})$ and $c_{kj}^\mathtt{bf} (\check{\bm{\alpha}})$ are defined as
\begin{align}
u_k^\mathtt{bf}(\check{\bm{\alpha}})\triangleq 
\begin{cases}
(N_\mathtt{t}\hat{\varrho}_k)^2 , & \text{for}\,\,\mathtt{ bf = MRT},\\
 1  , & \text{for}\,\,\mathtt{ bf = ZF},
\end{cases}
\end{align}
and
\begin{align}
\hspace{-0.2cm}c_{kj}^\mathtt{bf}(\check{\bm{\alpha}})\triangleq \!
\begin{cases}
\!N_\mathtt{t}\varrho_k \hat{\varrho}_j\!+(1-\delta_{kj})N^2_\mathtt{t}\hat{\varrho}_k
\hat{\varrho}_j,\!\! & \!\text{for}\,\,\mathtt{ bf \!= MRT},\\
\frac{{\tilde{\varrho}}_k {{\hat{\varrho}_j}^{-1}}}{(N_\mathtt{t}-K)}, & \!\text{for}\,\,\mathtt{ bf \!= ZF},
\end{cases}
\end{align}
{respectively, where $\delta_{kj}=1$ if $k=j$ and $0$ otherwise.
}

\end{theorem}

\begin{proof}
 See Appendix \ref{sec:sum_rate}. \hfill 
\qed
\end{proof}
{The rate in \eqref{rate_final}  depends on the HRIS coefficients through the squared magnitudes $\check{\alpha}_n=|\alpha_n|^2$. This follows from the MMSE channel estimate in Theorem~\ref{thm:mmse_estimate}, whose
statistical characterization depends on the covariance terms
$\varrho_k$ in \eqref{eq:cov_gk} and $\Psi_k$ in \eqref{psi_k}. Since these terms depend on $|\alpha_n|^2$, the resulting closed-form rate
expression is governed by the amplitude gains provided by the HRIS
coefficients.}
\begin{remark}
    {The sum rate of the considered system under the availability of perfect CSI can still be  obtained from \eqref{rate_final} as a special case by
setting $\hat{\varrho}_k=\varrho_k$ and
$\tilde{\varrho}_k= 0$, $\forall k$ in 
$u_k^{\mathtt{bf}}(\check{\boldsymbol{\alpha}})$,
$c_{kj}^{\mathtt{bf}}(\check{\boldsymbol{\alpha}})$, and
$d_k(\check{\boldsymbol{\alpha}})$. Furthermore, the term
$(1-\delta_{kj})N_\mathtt{t}^2\hat{\varrho}_k\hat{\varrho}_j$ in
$c_{kj}^{\mathtt{MRT}}(\check{\boldsymbol{\alpha}})$, which arises due
to pilot contamination, is absent under perfect CSI.}
\end{remark}

Based on \eqref{rate_final}, we investigate the transmit power scaling laws as a function of $N$. To this end, we consider the simplified case of equal-power allocation between communication and sensing and among all the communication users, i.e, $P_{\mathtt{c}}^{\mathtt{bf}}=P_{\mathtt{s}}=P_\mathtt{tot}^{\mathtt{bf}}/2$ and $\xi_1 = \xi_2 = \cdots = \xi_K$. The power scaling laws for this case are provided in the following remark.

\begin{remark}
\label{rem:sum_rate_scaling}
  For large $N$, 
  under the power scaling
$P_\mathsf{tot}^\mathtt{bf}=\frac{e_\mathtt{t}}{N^\varepsilon}$, where  $e_\mathsf{t}$ is a constant and $\varepsilon>0$ is a scaling exponent, the achievable rate in \eqref{rate_final} admits the following asymptotic form 
\begin{align}
\bar{\mathcal{R}}^\mathtt{bf}_k\approx  \tau_\mathtt{0}\log_2(1+\bar{\varsigma}^{\mathtt{bf}}_k),    \label{rate_asymptotic}
\end{align}
with 
\begin{align}
\label{sinr_approx}
\hspace{-0.15cm}\bar{\varsigma}^{\mathtt{bf}}_k\!\approx\!\!
\begin{cases}
\displaystyle
\!\!\frac{N_\mathtt{t}\hat{o}_k^{2}e_\mathtt{t} N^{(1-\varepsilon)}}
{2\left(\varkappa_k e_\mathtt{t}N^{(1-\varepsilon)}+\sigma_\mathtt{z}^{2}\right)\sum_{j=1}^{K}\hat{o}_j}
, & \hspace{-0.3cm}\text{for}\,\,\mathtt{ bf \!= MRT},\\
\displaystyle
\!\!\frac{(N_\mathtt{t}-K)e_\mathtt{t} N^{1-\varepsilon}}
{\left(\tilde{o}_k e_\mathtt{t} N^{(1-\varepsilon)}+2\sigma_\mathtt{z}^{2}\right)\sum_{j=1}^{K}\hat{o}_j^{-1}}
, & \!\!\!\!\!\text{for}\,\,\mathtt{ bf\! = ZF},
\end{cases}
\end{align}
where $\hat{o}_k\triangleq {\varkappa_k}-\tilde{\varkappa}_k$ and $\tilde{o}_k\triangleq \varkappa_k+\tilde{\varkappa}_k$ with pilot contamination, while $\hat{o}_k= \tilde{o}_k = \varkappa_k$ without pilot contamination. Here, $\varkappa_k \triangleq\beta_{\mathtt{UR}_k}\beta_{\mathtt{RB}}$ and $\tilde{\varkappa}_k\triangleq \frac{\sum_{j \in \mathcal{P}_k \backslash \{k\} } \varkappa_k}{1+ \frac{\sum_{j \in \mathcal{P}_k \backslash \{k\} }\beta_{\mathtt{UR}_j}}{\beta_{\mathtt{UR}_k}}}$.
It is observed from \eqref{sinr_approx} that the asymptotic rate converges to a non-zero value for any $\varepsilon \leq 1$, showing that the HRIS can maintain non-vanishing rates even at very low transmit power. 
Interestingly, \eqref{sinr_approx} also reveals that, in the large $N$ regime, the power-scaling laws are independent of pilot contamination. {For the special case of a fully ARIS, i.e., $N_\mathsf{a}=N$, it follows from \eqref{aggregate_noise} that, for large $N$ and unconstrained transmit power at the ARIS, we have $\sigma_\mathtt{z}^2=\mathcal{O}(N)$. Substituting $\sigma_\mathtt{z}^2=\mathcal{O}(N)$ into \eqref{sinr_approx} yields $\bar{\varsigma}_k^{\mathtt{bf}}=\mathcal{O}(N^{-\varepsilon})$. Thus, for any $\varepsilon>0$,
$\bar{\varsigma}_{k}^{\mathtt{bf}}\to 0$ as $N\to\infty$, and the
corresponding achievable rate vanishes. Although derived for the unconstrained case, this power-scaling law also
holds under the constrained transmit-power setting, as will be justified in
Section~\ref{sec:simulation_results}.}

\end{remark}

\begin{proof}
    See Appendix~\ref{sec:rate_scaling_law}.
\hfill 
\qed
\end{proof}
These observations remain generally valid for unequal power allocation case, as will be numerically justified in Section~\ref{sec:simulation_results}.

\subsection{Sensing CRLB}

We use the CRLB associated with $(\theta, \phi)$ as the
sensing performance metric, which is obtained from the
inverse of the Fisher information matrix (FIM). To derive the
FIM, we rewrite \eqref{echo} as
\begin{align}
 \mathbf {y}_\mathtt{s} = \mathtt {vec}(\mathbf {Y}_\mathtt{s}) = \mathbf {v} + \mathbf {n}_\mathtt{s},   
\end{align}
where $\mathbf {v}=\beta_\mathtt{s}\mathtt{vec}(\mathbf{A}\mathbf{X})$ and $\mathbf {n}_\mathtt{s}=\mathsf{vec}(\mathbf{N}_\mathtt{s})$.
Denote by $\boldsymbol{\zeta}=[\theta,\phi,\tilde{\boldsymbol{\beta}}_\mathtt{s}]^\T \in \mathbb{R}^{4\times 1}$ the unknown parameters to be estimated, where $\tilde{\boldsymbol{\beta}}_\mathtt{s}=[\mathtt{Re}(\beta_\mathtt{s}),\mathtt{Im}(\beta_\mathtt{s})]^\T$. 
Then, 
 the FIM for
estimating ${\boldsymbol{\zeta}}$ from $\mathbf {y}_\mathtt{s}$ is given as \cite{x.song}  
\begin{align} \label{fim_basic}
\mathbf{T}_{\boldsymbol{\zeta}} = \frac {2L}{\sigma _\mathtt{s}^{2}}\begin{bmatrix}
{T}_{\theta\theta} & {T}_{\theta\phi} & \mathbf{t}_{\theta\tilde{\boldsymbol{\beta}}_\mathtt{s}} \\
{T}_{\theta\phi} & {T}_{\phi\phi} & \mathbf{t}_{\phi\tilde{\boldsymbol{\beta}}_\mathtt{s}} \\
\mathbf{t}^\T_{\theta\tilde{\boldsymbol{\beta}}_\mathtt{s}} & \mathbf{t}^\T_{\phi\tilde{\boldsymbol{\beta}}_\mathtt{s}} & \mathbf{T}_{\tilde{\boldsymbol{\beta}}_\mathtt{s}\tilde{\boldsymbol{\beta}}_\mathtt{s}}
\end{bmatrix}.
\end{align}
{From~\eqref{fim_basic}, the equivalent FIM for $(\theta,\phi)$ 
is obtained via the Schur complement with respect to
$\tilde{\boldsymbol{\beta}}_{\mathtt{s}}$ as
\begin{align}
\mT_{\theta,\phi}
&=
\frac{2L}{\sigma_{\mathtt{s}}^{2}}
\begin{bmatrix}
\Jtt-\Jta\Jaa^{-1}\Jta^\T
&
\Jtp-\Jta\Jaa^{-1}\Jpa^\T
\\
\Jtp-\Jpa\Jaa^{-1}\Jta^\T
&
\Jpp-\Jpa\Jaa^{-1}\Jpa^\T
\end{bmatrix},
\label{eq:efim}
\end{align}
where we have defined
\begin{align}
\label{eq:tpp}
T_{\psi\psi}
&=
|\beta_{\mathtt{s}}|^2
\mathtt{tr}\left(
\dot{\mA}_{\psi}\mR_\mathbf{x}^{\mathtt{bf}}\dot{\mA}_{\psi}^{\H}
\right)  \in \mathbb{R}^{1\times 1},
\\
T_{\theta\phi}
&=
|\beta_{\mathtt{s}}|^2
\mathtt{tr}\left(
\dot{\mA}_{\theta}\mR_\mathbf{x}^{\mathtt{bf}}\dot{\mA}_{\phi}^{\H}
\right)  \in \mathbb{R}^{1\times 1},
\\
\mathbf{t}_{\psi\tilde{\boldsymbol{\beta}}_{\mathtt{s}}}
&=
\mathtt{Re}\left\{
\beta_{\mathtt{s}}^{*}
\mathtt{tr}\left(
\mA\mR_\mathbf{x}^{\mathtt{bf}}\dot{\mA}_{\psi}^{\H}
\right)[1,j]
\right\} \in \mathbb{R}^{1\times 2},
\\
\mathbf{T}_{\tilde{\boldsymbol{\beta}}_{\mathtt{s}}
\tilde{\boldsymbol{\beta}}_{\mathtt{s}}}
&=
\mathtt{tr}\!\left(
\mA\mR_\mathbf{x}^{\mathtt{bf}}\mA^{\H}
\right)\mI_2 \in \mathbb{R}^{2\times 2}.
\label{eq:tbb}
\end{align}
Here, $\psi \in \{\theta,\phi\}$, 
$\dot{\mA}_{\theta} = \bdottheta \va^\H + \vb \adottheta^\H,\ 
		\dot{\mA}_{\phi} = \bdotphi \va^\H + \vb \adotphi^\H$, where $\dot{\mathbf{a}}_\psi$ and $\dot{\mathbf{b}}_\psi$ can be computed using \eqref{app:a_dot_theta}--\eqref{app:a_dot_phi}. The closed-form expression for
$\mR_\mathbf{x}^{\mathtt{bf}}$ follows from Appendix~\ref{sec:transmit_covaraince} as
\begin{align}
\mR_\mathrm{x}^{\mathtt{bf}}
&=
\sum_{k=1}^{K}
\xi_k\nu_k^{\mathtt{bf}}\mI_{N_{\mathtt t}}
+
\eta\vu\vu^{\H}.
\label{eq:Rx_bf}
\end{align}
Here, $\mathtt{bf}\in\{\mathtt{MRT},\mathtt{ZF}\}$,
$\eta\triangleq\|\boldsymbol{\rho}\|^{2}$, and
$\nu_k^{\mathtt{bf}}$ is given in \eqref{nu}. 
From~\eqref{eq:efim}, the CRLBs associated with $\theta$ and $\phi$ are
obtained from the diagonal entries of
$\mT_{\theta,\phi}^{-1}$ as
\begin{align}
\widehat{\mathtt{CRLB}}_{\theta}
&=
\frac{\sigma_{\mathtt{s}}^{2}}{2L}
\left(
T_{\theta\theta}
-
\widetilde{T}_{\theta\theta}
-
\frac{
\left(
T_{\theta\phi}
-
\widetilde{T}_{\theta\phi}
\right)^2
}{
T_{\phi\phi}
-
\widetilde{T}_{\phi\phi}
}
\right)^{-1},
\label{eq:general_crlb_theta}\\
\widehat{\mathtt{CRLB}}_{\phi}
&=
\frac{\sigma_{\mathtt{s}}^{2}}{2L}
\left(
T_{\phi\phi}
-
\widetilde{T}_{\phi\phi}
-
\frac{
\left(
T_{\theta\phi}
-
\widetilde{T}_{\theta\phi}
\right)^2
}{
T_{\theta\theta}
-
\widetilde{T}_{\theta\theta}
}
\right)^{-1},
\label{eq:general_crlb_phi}
\end{align}
where
\begin{align}
\widetilde{T}_{\theta\theta}
&\triangleq
\mathbf{t}_{\theta\tilde{\boldsymbol{\beta}}_{\mathtt{s}}}
\mathbf{T}_{\tilde{\boldsymbol{\beta}}_{\mathtt{s}}
\tilde{\boldsymbol{\beta}}_{\mathtt{s}}}^{-1}
\mathbf{t}_{\theta\tilde{\boldsymbol{\beta}}_{\mathtt{s}}}^{\T},
\label{eq:tilde_T_theta_theta}\\
\widetilde{T}_{\theta\phi}
&\triangleq
\mathbf{t}_{\theta\tilde{\boldsymbol{\beta}}_{\mathtt{s}}}
\mathbf{T}_{\tilde{\boldsymbol{\beta}}_{\mathtt{s}}
\tilde{\boldsymbol{\beta}}_{\mathtt{s}}}^{-1}
\mathbf{t}_{\phi\tilde{\boldsymbol{\beta}}_{\mathtt{s}}}^{\T},
\label{eq:tilde_T_theta_phi}\\
\widetilde{T}_{\phi\phi}
&\triangleq
\mathbf{t}_{\phi\tilde{\boldsymbol{\beta}}_{\mathtt{s}}}
\mathbf{T}_{\tilde{\boldsymbol{\beta}}_{\mathtt{s}}
\tilde{\boldsymbol{\beta}}_{\mathtt{s}}}^{-1}
\mathbf{t}_{\phi\tilde{\boldsymbol{\beta}}_{\mathtt{s}}}^{\T}.
\label{eq:tilde_T_phi_phi}
\end{align}
{Based on the CRLB expressions in
\eqref{eq:general_crlb_theta}--\eqref{eq:general_crlb_phi} and the
transmit covariance matrix in \eqref{eq:Rx_bf}, we evaluate the CRLBs
for the misaligned sensing beamformer
$\mathbf{u}=\mathbf{a}(\hat{\theta},\hat{\phi})$ in
Fig.~\ref{fig:crlb_error}. The CRLBs for both azimuth and
elevation estimation increase with the target-angle error. However,
a $5^\circ$ increase in the angle error results in an increase of only
approximately $3.5$~dB. Thus, the sensing CRLB is
relatively insensitive to misalignment of the sensing beam.}}
\begin{figure}[!t]
	\centering
\includegraphics[scale=0.46]{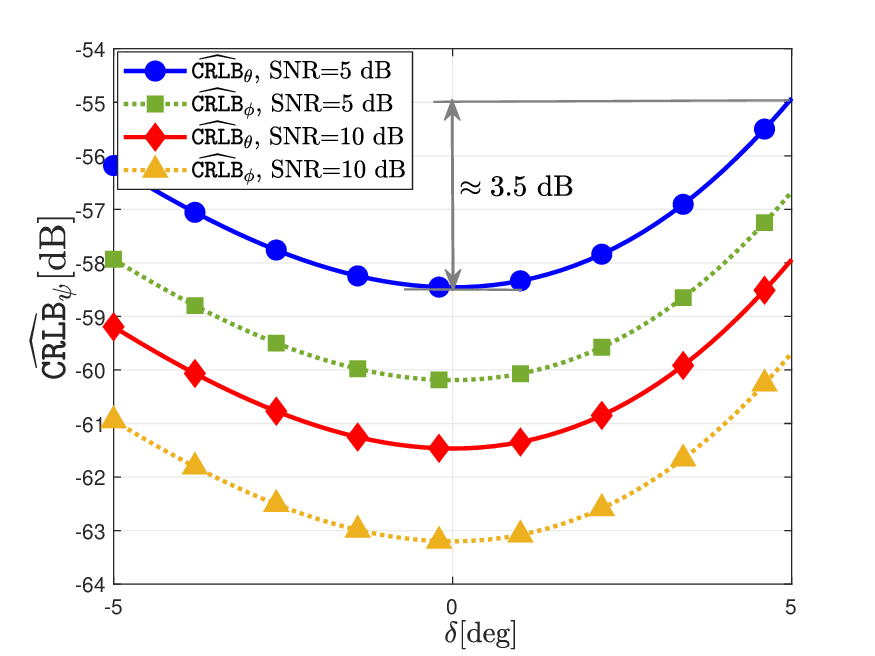}
	\caption{\small{{CRLB versus target-angle error $\delta$.
    }}}
	\label{fig:crlb_error}
\end{figure}

{The general CRLB expressions in
\eqref{eq:general_crlb_theta}--\eqref{eq:general_crlb_phi} are
nevertheless difficult to handle analytically due to the intricate
structure of the Schur-complement terms
$\widetilde{T}_{\theta\theta}$, $\widetilde{T}_{\theta\phi}$, and
$\widetilde{T}_{\phi\phi}$. 
This makes them intractable for analytical characterization and, even
more so, for optimization-based system design.
To obtain tractable closed-form CRLB expressions that facilitate the
subsequent analysis and design, we consider the case $(\hat{\theta}, \hat{\phi}) = ({\theta}, {\phi})$ and derive the simplified closed-form CRLB
expressions in the following theorem.}

\begin{theorem} \label{thm:CRLB}
A unified closed-form expression for the CRLB associated with $\theta$ and $\phi$, employing the MRT and ZF precoder is given by \eqref{CRLB_psi}, where $\psi \in\{\theta,\phi\}$, $\mathtt{bf} \in \{\mathtt{MRT},\mathtt{ZF}\}$,
\begin{figure*}[b]
\hrulefill
\begin{align}
\label{CRLB_psi}
\hspace{-0.2cm}\mathtt{CRLB}^\mathtt{bf}_\psi(\bm{\xi},\eta,\check{\bm{\alpha}}) =
\frac{6\sigma_\mathtt{s}^2 \left(\boldsymbol{\xi}^\T\boldsymbol{\nu}^\mathtt{bf}(\check{\bm{\alpha}})   \varpi_{\psi} 
    + \eta N_{\mathtt{t}} \varphi_{\psi}\right)}
{{L|\beta_\mathtt{s}|^2 N_{\mathtt{t}} N_{\mathtt{r}} \pi^2}
\left(
\left(\boldsymbol{\xi}^\T\boldsymbol{\nu}^\mathtt{bf}(\check{\bm{\alpha}}) \varpi_{\theta} 
    + \eta N_{\mathtt{t}} \varphi_{\theta}\right)
\left(\boldsymbol{\xi}^\T\boldsymbol{\nu}^\mathtt{bf}(\check{\bm{\alpha}}) \varpi_{\phi} 
    + \eta N_{\mathtt{t}} \varphi_{\phi}\right)
- \left(\boldsymbol{\xi}^\T\boldsymbol{\nu}^\mathtt{bf}(\check{\bm{\alpha}}) \varpi_{\theta\phi} 
    + \eta N_{\mathtt{t}} \varphi_{\theta\phi}\right)^2
\right)}.
\end{align}
\end{figure*}
   ${  \bm{\nu}}^\mathtt{bf}(\check{\bm{\alpha}}) \triangleq\left[{\nu_1^\mathtt{bf}(\check{\bm{\alpha}}}),\dots,{\nu^\mathtt{bf}_K(\check{\bm{\alpha}})}\right]^\T$ with ${{{\nu}^\mathtt{bf}_k}}(\check{\bm{\alpha}})$ given in \eqref{nu},  and
\begin{align}
\hspace{-0.2cm}\varpi_{\theta} &\!\triangleq\! \left(N_{\mathtt{ty}}^2 + N_{\mathsf{ry}}^2 - 2\right)\cos^2(\theta)\sin^2(\phi), \\
\varphi_{\theta} &\!\triangleq\! \left(N_{\mathsf{ry}}^2 - 1\right)\cos^2(\theta)\sin^2(\phi), \\
\varpi_{\phi} &\!\triangleq\!\left(\left(N_{\mathtt{ty}}^2 \!+\! N_{\mathsf{ry}}^2 \!- 2\right)\!\sin^2(\theta) 
    \!+ \!\left(N_{\mathtt{tz}}^2 + N_{\mathsf{rz}}^2 \!-\! 2\right)\right)\!\cos^2(\phi), \\
\varphi_{\phi} &\!\triangleq \left(\left(N_{\mathsf{ry}}^2 - 1\right)\sin^2(\theta) 
    + \left(N_{\mathsf{rz}}^2 - 1\right)\right)\cos^2(\phi), \\
\varpi_{\theta\phi} &\!\triangleq \frac{1}{4}\left(N_{\mathtt{ty}}^2 + N_{\mathtt{tz}}^2 - 2\right)\sin(2\theta)\sin(2\phi), \\
\varphi_{\theta\phi} &\!\triangleq \frac{1}{4}{\left(N_{\mathsf{ry}}^2 - 1\right)}{\sin(2\theta)\sin(2\phi)}.
\end{align}
\end{theorem}

\begin{proof}
 See Appendix~\ref{sec:appendixe}.  
\hfill 
\qed
\end{proof}

\begin{remark}
It is worth noting from \eqref{CRLB_psi} that the CRLB depends on the HRIS coefficients through $\bm{\xi}^\T {\bm{\nu}^\mathtt{bf}(\check{\bm{\alpha}})}$. This shows that, although the HRIS does not directly assist the sensing function, it still affects the sensing performance. {With equal power allocation between communications and sensing, i.e., $\boldsymbol{\xi}^\T\boldsymbol{\nu}^\mathtt{bf}(\check{\boldsymbol{\alpha}})=\eta=\frac{P_\mathtt{tot}^\mathtt{bf}}{2N_\mathtt{t}}$, the CRLB becomes independent of the HRIS coefficients, showing that the sensing benefits of the HRIS cannot be exploited under this simplified allocation. This motivates us to consider the joint optimization of the power allocation and HRIS coefficients, as elaborated in the next section.}
\end{remark}


 \section{Problem Formulation and Solution}
\label{sec:problem_formulation}

Our goal is to maximize the communication sum-rate in \eqref{rate_final}, subject to the constraints on the sensing CRLB, the total transmit power, and HRIS coefficients. 
Based on the derived analytical expressions for the performance metrics in Section~\ref{sec:performance_analysis}, we note that the sum-rate and CRLB are affected by the total power for sensing, i.e., $\eta=\|\bm{\rho}\|^2$, rather than $\{\rho_k\}_{k=1}^K$. As a result, the problem of optimizing $\{\xi_k,\rho_k\}_{k=1}^K$ reduces to optimizing the variables $\{\xi_k\}_{k=1}^K$ and $\eta$.  Furthermore, recalling that both the sum-rate and the CRLB depend on the squared magnitudes of the HRIS coefficients, i.e., $\check{\bm{\alpha}}=\left[\check{\alpha}_1,\cdots, \check{\alpha}_N\right]^\T$, where $\check{\alpha}_n=|\alpha_n|^2$, we focus on optimizing the variable $\check{\bm{\alpha}}$.
Consequently, the sum-rate maximization problem can be formulated as follows 
\begin{subequations}\label{eq:main_problem}
\begin{align}
\qquad
&\underset{\bm{\xi},{\eta},\check{\bm{\alpha}}}
  {\text{maximize}} \sum_{k=1}^K \mathcal{R}^\mathtt{bf}_k (\bm{\xi},{\eta},\check{\bm{\alpha}}) \label{problem1} \\
& \hspace{0.5cm}{\text{s.t.}} \quad \mathtt{CRLB}_\theta(\bm{\xi},\eta,\check{\bm{\alpha}})\leq \mathtt{CRLB}_\theta^\mathsf{th},\label{CRLB_theta_constraint}\\
& \hspace{1.4cm}\mathtt{CRLB}_\phi(\bm{\xi},\eta,\check{\bm{\alpha}}) \leq \mathtt{CRLB}_\phi^\mathsf{th}, \label{CRLB_phi_constraint} \\
& \hspace{1.45cm}\sum_{k=1}^K{N_\mathtt{t}}\xi_k\nu_k^\mathtt{bf}(\check{\bm{\alpha}})+ N_\mathtt{t}\eta \leq P_\mathtt{t}, \label{constraint_transmit power}\\
& \hspace{1.45cm}
{\check{\alpha}_n=1, n \notin \mathcal{I}_\mathsf{a} },
\label{hris_constaints}\\
& \hspace{1.45cm}
\check{\alpha}_n\leq \alpha_{\max}^2, n \in \mathcal{I}_\mathsf{a} \label{hris_constraint1},
\end{align}
\end{subequations}
where $\mathtt{CRLB}_\theta$ and $\mathtt{CRLB}_\phi$ are the predefined CRLB thresholds.
The objective function \eqref{problem1} is nonconcave with respect to $\bm{\xi}$, $\eta$, and $\bm{\check{\alpha}}$.  {The CRLB constraints in \eqref{CRLB_theta_constraint} and
\eqref{CRLB_phi_constraint} jointly govern the azimuth- and
elevation-angle estimation accuracy and have a highly intricate structure. Moreover, the variables $(\boldsymbol{\xi},\eta)$ and $\bm{\check{\alpha}}$ are coupled in  \eqref{CRLB_theta_constraint}--\eqref{constraint_transmit power}.} Hence, \eqref{eq:main_problem} is non-trivial to solve. To overcome this, we propose an AO algorithm that  iteratively solves the power-allocation and HRIS-beamforming subproblems.
{The problem in \eqref{eq:main_problem} is formulated using the
CRLB expressions in Theorem~\ref{thm:CRLB} without target-angle errors. Accordingly, the
CRLB constraints in \eqref{CRLB_theta_constraint} and
\eqref{CRLB_phi_constraint} are evaluated for the aligned sensing
beamformer. Although a small target-angle error degrades the sensing
performance, it is expected to have only a limited effect on the
communications sum-rate objective. This is because the main effect of sensing on the
optimization arises from satisfying the CRLB constraints rather than from
the precise target angles at which they are evaluated.}

\subsection{Power Allocation}
\label{power_allocaion}

We first focus on the subproblem of optimizing the power allocation coefficients, i.e., $\{\xi_k\}_{k=1}^{K}$ and $\eta$
for a given $\check{\bm{\alpha}}$. To this end, we employ an SCA framework, which sequentially approximates the nonconvex objective function and constraints in \eqref{eq:main_problem} with convex surrogate functions, as detailed below.

\subsubsection{Transformation of Objective Function \eqref{problem1}}

We first rewrite \eqref{rate_final} as $\mathcal{R}_k^{\mathtt{bf}}(\bm{\xi},\eta)=  {\tau_\mathtt{0}}\log_2 \left(1+ \frac{\mathcal{R}^\mathtt{bf}_{\mathtt{Nu},k}(\bm{\xi}) }{\mathcal{R}^\mathtt{bf}_{\mathtt{De},k}(\bm{\xi},{\eta})}\right)$, where $\mathcal{R}^\mathtt{bf}_{\mathtt{Nu},k}(\bm{\xi})\triangleq {u_k^{\mathtt{bf}}\xi_k}$
and $\mathcal{R}^\mathtt{bf}_{\mathtt{De},k}(\bm{\xi},{\eta})\triangleq{\eta d_k  + \sum _{j=1}^{K} \xi_j c_{kj}^{\mathtt{bf}} + \sigma_{\mathtt{z}}^2}$. Note that both ${\mathcal{R}^\mathtt{bf}_{\mathtt{Nu},k}(\bm{\xi}) }$ and $\mathcal{R}^\mathtt{bf}_{\mathtt{De},k}(\bm{\xi},{\eta})$ are linear functions of $(\bm{\xi},{\eta})$. Thus, by using \cite[(22)]{concave_lower_twc}, we obtain a concave lower bound of $\mathcal{R}^\mathtt{bf}_k (\bm{\xi},{\eta})$ around the feasible point $(\bm{\xi}^{[i]},\eta^{[i]})$ at the $i$-th iteration as 
\begin{align}
\mathcal{R}^\mathtt{bf}_k (\bm{\xi},{\eta}) &\geq  {\frac{{\tau_\mathtt{0}}}{\log2}\left(q_{\mathtt{u},k}^{[i]}- \frac{q_{\mathtt{v},k}^{[i]}}{\mathcal{R}^\mathtt{bf}_{\mathtt{Nu},k}(\bm{\xi})}- q_{\mathtt{w},k}^{[i]} \mathcal{R}^\mathtt{bf}_{\mathtt{De},k}(\bm{\xi},{\eta})\right)}\nonumber \\
  &\triangleq \tilde{\mathcal{R}}^\mathtt{bf}_k (\bm{\xi}^{[i]},{\eta}^{[i]}),
\end{align}
where 
\begin{align}
q_{\mathtt{u},k}^{[i]} &\triangleq \!\log \!\left(\!\!1\!+ \!\frac{\mathcal{R}^\mathtt{bf}_{\mathtt{Nu},k}(\bm{\xi}^{[i]}) }{\mathcal{R}^\mathtt{bf}_{\mathtt{De},k}(\bm{\xi}^{[i]},{\eta}^{[i]}) }\!\right)\!\!+\!\frac{2\mathcal{R}^\mathtt{bf}_{\mathtt{Nu},k}(\bm{\xi}^{[i]}) }{\!\mathcal{R}^\mathtt{bf}_{\mathtt{Nu},k}\!(\!\bm{\xi}^{[i]}\!) \!+\!\mathcal{R}^\mathtt{bf}_{\mathtt{De},k}\!(\!\bm{\xi}^{[i]},\!{\eta}^{[i]}) },\\
 q_{\mathtt{v},k}^{[i]} &\triangleq \frac{\left(\mathcal{R}^\mathtt{bf}_{\mathtt{Nu},k}(\bm{\xi}^{[i]}) \right)^2}{\mathcal{R}^\mathtt{bf}_{\mathtt{Nu},k}(\bm{\xi}^{[i]}) +\mathcal{R}^\mathtt{bf}_{\mathtt{De},k}(\bm{\xi}^{[i]},{\eta}^{[i]}) },\\
 q_{\mathtt{w},k}^{[i]} &\!\triangleq \frac{\mathcal{R}^\mathtt{bf}_{\mathtt{Nu},k}(\bm{\xi}^{[i]}) }{\left(\mathcal{R}^\mathtt{bf}_{\mathtt{Nu},k}(\bm{\xi}^{[i]}) \!+\!\mathcal{R}^\mathtt{bf}_{\mathtt{De},k}(\bm{\xi}^{[i]},{\eta}^{[i]}) \right)\mathcal{R}^\mathtt{bf}_{\mathtt{De},k}(\bm{\xi}^{[i]},\!{\eta}^{[i]})}.
\end{align}

\subsubsection{Transformation of Constraints \eqref{CRLB_theta_constraint} and \eqref{CRLB_phi_constraint}}

We recast \eqref{CRLB_theta_constraint} and \eqref{CRLB_phi_constraint} as second-order cone (SOC) constraints. 
Specifically, denote  $\Delta_\theta \triangleq\boldsymbol{\xi}^\T\bm{\nu}^{\mathtt{bf}}\varpi_{\theta}
      + \eta\,N_{\mathtt{t}}\varphi_{\theta} $ and $\Delta_\phi \triangleq{\boldsymbol{\xi}^\T\bm{\nu}^{\mathtt{bf}}\varpi_{\phi}
      + \eta\,N_{\mathtt{t}}\varphi_{\phi}}$, then \eqref{CRLB_theta_constraint} can be rewritten as 
\begin{align}
 \frac{
\left(\boldsymbol{\xi}^\T\bm{\nu}^{\mathtt{bf}}\varpi_{\theta\phi}
      + \eta N_{\mathtt{t}}\varphi_{\theta\phi}\right)^{2}}
{\Delta_\phi}
&\leq
\Delta_\theta -
\frac{1}{\lambda_\mathtt{0}\mathtt{CRLB}^\mathsf{th}_\theta},
\label{constraint_CRLB_theta_1}
\end{align}
where ${\lambda}_\mathtt{0}\triangleq \frac{{L|\beta_{\mathtt{s}}|^{2}N_{\mathtt{t}}N_{\mathtt{r}}\pi^{2}}}{{6\sigma_{\mathtt{s}}^{2}}}$. By employing the SOC constraint transformation, \eqref{constraint_CRLB_theta_1} can be equivalently written as
\begin{align}
\label{soc_constraint_theta}
& \hspace{-0.1cm} \left\lVert \left[2{
\left(\boldsymbol{\xi}^\T\bm{\nu}^{\mathtt{bf}}\varpi_{\theta\phi}
      + \eta N_{\mathtt{t}}\varphi_{\theta\phi}\right)};\Delta_\phi-\left(\Delta_\theta-\frac{1}{\lambda_\mathtt{0}\mathtt{CRLB}^\mathsf{th}_\theta}\right)\right]\right\rVert_2 \nonumber \\
      &\hspace{3.5cm}\leq \Delta_\phi+\left(\Delta_\theta-\frac{1}{\lambda_\mathtt{0}\mathtt{CRLB}^\mathsf{th}_\theta}\right).
\end{align}
Likewise, we recast \eqref{CRLB_phi_constraint} as
\begin{align}
\label{soc_constraint_phi}
& \hspace{-0.1cm} \left\lVert \left[2{
\left(\boldsymbol{\xi}^\T\bm{\nu}^{\mathtt{bf}}\varpi_{\theta\phi}
      + \eta N_{\mathtt{t}}\varphi_{\theta\phi}\right)};\Delta_\theta-\left(\Delta_\phi-\frac{1}{\lambda_\mathtt{0}\mathtt{CRLB}^\mathsf{th}_\phi}\right)\right]\right\rVert_2 \nonumber \\
      &\hspace{3.5cm}\leq \Delta_\theta+\left(\Delta_\phi-\frac{1}{\lambda_\mathtt{0}\mathtt{CRLB}^\mathsf{th}_\phi}\right).
\end{align}
Thus, the approximate convex
problem for \eqref{problem1} solved at $i$-the iteration is given as
\begin{align}
    \label{power_allocation_convex_final}\underset{\bm{\xi},{\eta}}
  {\text{maximize}} \!\sum_{k=1}^K \tilde{\mathcal{R}}^\mathtt{bf}_k (\bm{\xi},\eta;\bm{\xi}^{[i]},{\eta}^{[i]}), \, \text{s.t.} \,\,\,\eqref{constraint_transmit power},\eqref{soc_constraint_theta}, \eqref{soc_constraint_phi}.  
\end{align}

\subsection{HRIS Beamforming}
\label{hris_beamforming}

After obtaining the power allocation coefficients $(\bm{\xi},\eta)$, we aim at solving the HRIS coefficient gains $\check{\bm{\alpha}}$.
\subsubsection{Transformation of Objective Function \eqref{problem1}}
We apply the FP framework to \eqref{problem1}, which yields the following lemma.

\begin{lemma}
\label{lemma_FP}
 By introducing the auxiliary variables $\mathbf{t}_\mathtt{u}^\mathtt{bf}=\{t^\mathtt{bf}_{\mathtt{u}1},\cdots,t^\mathtt{bf}_{\mathtt{u}K}\}$, $\mathbf{t}_\mathtt{v}^\mathtt{bf}=\{t^\mathtt{bf}_{\mathtt{v}1},\cdots,t^\mathtt{bf}_{\mathtt{v}K}\}$, and $\mathbf{t}_\mathtt{w}=\{t_{\mathtt{w}1},\cdots,t_{\mathtt{w}K}\}$,  \eqref{problem1} can be equivalently reformulated as
    \begin{align}
\underset{\check{\boldsymbol{\alpha}}, \mathbf{t}_\mathtt{u}^\mathtt{bf},\mathbf{t}_\mathtt{v}^\mathtt{bf},\mathbf{t}_\mathtt{w}}
  {\text{maximize}} \sum_{k=1}^K \frac{{\tau_\mathtt{0}}}{\log(2)}  
{f}^\mathtt{bf}_k(\mathbf{\Omega}), \quad \text{s.t.} \, \eqref{CRLB_theta_constraint}- \eqref{hris_constraint1},
\label{lem:objecive_hris}
    \end{align}
where $\mathtt{bf}\in \{\mathtt{MRT},\mathtt{ZF}\}$,  $\mathbf{\Omega}\triangleq\left\{\check{\bm{\alpha}},t_{\mathtt{u}k}^\mathtt{bf},t_{\mathtt{v}k}^\mathtt{bf},t_{\mathtt{w}k}\right\}_{\forall k}$, and 
    \begin{align}
       {f}^\mathtt{bf}_k(\mathbf{\Omega}) &\triangleq  \log\left(1+t_{\mathtt{u}k}^\mathtt{bf}\right)\!-t_{\mathtt{u}k}^\mathtt{bf}\!+2t_{\mathtt{v}k}^\mathtt{bf}\sqrt{(1+t_{\mathtt{u}k}^\mathtt{bf}){u_k^{\mathtt{bf}}(\check{\bm{\alpha}})}\xi_k} \nonumber \\
      & \hspace{0.2cm} -(t_{\mathtt{v}k}^\mathtt{bf})^2\bigg(\eta N_\mathtt{t}\varrho_k(\check{\bm{\alpha}})+N_\mathtt{t}\check{\varrho}_k(\check{\bm{\alpha}})\xi_k+\sum_{j=1}^K \xi_j c_{kj}^\mathtt{bf}(\check{\bm{\alpha}})\nonumber \\
      & \hspace{0.2cm}+\tilde{\sigma}_\mathtt{p}^2+\beta_\mathtt{RB}\sigma^2_\mathtt{a}\check{\bm{\alpha}}^\T\mathbf{e}\bigg).
      \label{obj_hris_recast_fig}
    \end{align}
Here, $\mathbf{e}\in \mathbb{C}^{N\times 1}$ is an all-ones column vector, 
\begin{align}
\tilde{\sigma}^2_\mathtt{p}&\triangleq \sigma^2_\mathtt{p}-\beta_{\mathtt{RB}}  \sigma^2_\mathtt{a}(N-N_\mathtt{a}),\\
c_{kj}^\mathtt{MRT}(\check{\bm{\alpha}})&\triangleq \begin{cases}
N_\mathtt{t}\varrho_k(\check{\bm{\alpha}})\check{\varrho}_j(\check{\bm{\alpha}}), & j=k,\\N_\mathtt{t}\check{\varrho}_j(\check{\bm{\alpha}})\Big(\varrho_k(\check{\bm{\alpha}})+N_\mathtt{t}\check{\varrho}_k(\check{\bm{\alpha}})\Big),& j \neq k,
\end{cases} \\
c_{kj}^\mathtt{ZF}(\check{\bm{\alpha}})&\triangleq\frac{\varrho_k(\check{\bm{\alpha}})-\check{\varrho}_k(\check{\bm{\alpha}})}{(N_\mathtt{t}-K)\check{\varrho}_j(\check{\bm{\alpha}})},
\end{align}
with 
\begin{align}
\varrho_k(\check{\bm{\alpha}})&\triangleq \!\beta_{\mathtt{UB}_k} +\beta_{\mathtt{UR}_k}\beta_{\mathtt{RB}} \check{\bm{\alpha}}^\T\mathbf{e}, \\
    \check{\varrho}_k(\check{\bm{\alpha}})&\!\triangleq \breve{t}_{\mathtt{w}k}+ \tilde{t}_{\mathtt{w}k}\check{\bm{\alpha}}^\T\mathbf{e} \label{r_k_check},
\\
\breve{t}_{\mathtt{w}k}&\!\triangleq\! 2\sqrt{\tau_\mathtt{p}p_\mathtt{p}} \beta_{\mathtt{UB}_k}t_{\mathtt{w}k}\!-\!t_{\mathtt{w}k}^2\left(\sum_{k \in \mathcal {P}_{k} }\!\tau_\mathtt{p}p_\mathtt{p}\beta_{\mathtt{UB}_k}\!\!+\!\tilde{\sigma}^2_\mathtt{p}\right), \\
\tilde{t}_{\mathtt{w}k}&\!\triangleq\! 2\sqrt{\tau_\mathtt{p}p_\mathtt{p}} \beta_{\mathtt{UB}_k}\beta_{\mathtt{RB}}t_{\mathtt{w}k}\!-\!t_{\mathtt{w}k}^2\!\left(\!\!\beta_{\mathtt{RB}}\sigma_\mathtt{a}^2\!\!+\!\!\!\sum_{k \in \mathcal {P}_{k} }\!\!\!{\tau_\mathtt{p}p_\mathtt{p}}\beta_{\mathtt{RB}}\!\right)\!.  
\end{align}
For a given $\check{\bm{\alpha}}$, the subproblems with respect to $\big\{t^\mathtt{bf}_{\mathtt{u}k},t^\mathtt{bf}_{\mathtt{v}k},t_{\mathtt{w}k}\big\}$ in \eqref{lem:objecive_hris} are unconstrained convex problems with optimal values
\begin{align}
 \label{aux:t_u}t_{\mathtt{u}k}^{\mathtt{bf}^\star}&=\varsigma_k^{\mathtt{bf}},  \\
 \label{aux:t_v}t_{\mathtt{v}k}^{\mathtt{bf}^\star}&=\frac{\sqrt{(1+t_{\mathtt{u}k}^\mathtt{bf}){u_k^{\mathtt{bf}}(\check{\boldsymbol\alpha})\xi_k}}}{{{u_k^{\mathtt{bf}}(\check{\boldsymbol\alpha})\xi_k}+\eta d_k (\check{\boldsymbol\alpha}) + \sum _{j=1}^{K} \xi_j c_{kj}^{\mathtt{bf}} (\check{\boldsymbol\alpha})+ \sigma_{\mathtt{z}}^2(\check{\boldsymbol\alpha})} }, \\
 \label{aux:t_w}t_{\mathtt{w}k}^{\star}&= \frac{\sqrt{\tau_\mathtt{p} p_\mathtt{p} }\varrho_k(\check{\boldsymbol\alpha})}{\sum _{k \in \mathcal {P}_{k} } \tau_\mathtt{p} p_\mathtt{p} {\varrho}_k(\check{\boldsymbol\alpha})+\sigma_\mathtt{z}^2(\check{\boldsymbol\alpha})}. 
\end{align}
\end{lemma}

\begin{proof}
   See Appendix~\ref{sec:appendixf}. 
   \hfill 
\qed
\end{proof}

{We next adopt a block
coordinated descent (BCD)} framework to optimize the four blocks in \eqref{lem:objecive_hris}, i.e., $\{\check{\bm{\alpha}}\}$, $\{\mathbf{t}^\mathtt{bf}_\mathtt{u}\}$, $\{\mathbf{t}^\mathtt{bf}_\mathtt{v}\}$, and $\{\mathbf{t}_\mathtt{w}\}$  alternatively.
Specifically, for a given $\{\check{\bm{\alpha}}\}$, the optimal values of the auxiliary variables $\{\mathbf{t}^\mathtt{bf}_\mathtt{u},\mathbf{t}^\mathtt{bf}_\mathtt{v},\mathbf{t}_\mathtt{w}\}$ can be obtained by \eqref{aux:t_u}--\eqref{aux:t_w}.
Subsequently, for a given $\{\mathbf{t}^\mathtt{bf}_\mathtt{u},\mathbf{t}^\mathtt{bf}_\mathtt{v},\mathbf{t}_\mathtt{w}\}$, the optimal $\check{\bm{\alpha}}$ can be obtained by solving the following problems corresponding to the MRT and ZF precoders, respectively
\begin{align}
\label{equivalent_mr}
   \underset{\check{\boldsymbol{\alpha}}}
  {\text{maximize}} \sum_{k=1}^K  \frac{{\tau_\mathtt{0}}}{\log(2)}
{f}^\mathtt{MRT}_k(\check{{\bm{\alpha}}}), \quad \text{s.t.} \, \eqref{CRLB_theta_constraint}-\eqref{hris_constraint1}, \\
\underset{\check{\boldsymbol{\alpha}}}
  {\text{maximize}} \sum_{k=1}^K \frac{{\tau_\mathtt{0}}}{\log(2)} 
{f}^\mathtt{ZF}_k(\check{\bm{\alpha}}), \quad \text{s.t.}\,\eqref{CRLB_theta_constraint}-\eqref{hris_constraint1}, 
    \label{equivalent_zf}
\end{align}
where ${f}^\mathtt{bf}_k(\check{\bm{\alpha}})$ is given in \eqref{obj_hris_recast_fig} for $\mathtt{bf} \in \{\mathtt{MRT},\mathtt{ZF}\}$.
We note that the equivalent sum-rate maximization problem  \eqref{equivalent_mr}, associated with the MRT precoder, is convex.
However, \eqref{equivalent_zf} for the ZF case, remains nonconvex due to the fractional term $c_{kj}^\mathtt{ZF}(\check{\bm{\alpha}})\triangleq \frac{{\varrho_k(\check{\bm{\alpha}})-\check{\varrho}_k(\check{\bm{\alpha}})}}{(N_\mathtt{t}-K)\check{\varrho}_j(\check{\bm{\alpha}})} $ in \eqref{obj_hris_recast_fig}. By the first-order Taylor approximation around the feasible point $\check{\bm{\alpha}}^{[j]}$ at the $j$-th iteration,
an affine approximation of $c_{kj}^\mathtt{ZF}(\check{\bm{\alpha}})$ is given by 
\begin{align}
c_{kj}^\mathtt{ZF}(\check{\bm{\alpha}})&\approx c_{kj}^\mathtt{ZF}(\check{\bm{\alpha}}^{[j]})+\lambda_{kj}^{[j]}\mathbf{e}^\T(\check{\bm{\alpha}}-\check{\bm{\alpha}}^{[j]}) \triangleq \tilde{c}_{kj}^{\mathtt{ZF}}(\check{\bm{\alpha}};\check{\bm{\alpha}}^{[j]}),
\label{c_kj_bound}
\end{align}
where $\lambda_{kj}^{[j]}\triangleq \frac{\breve{t}_{\mathtt{w}j}(\beta_{\mathtt{UR}_k}\beta_{\mathtt{RB}}-\tilde{t}_{\mathtt{w}k})-\tilde{t}_{\mathtt{w}j}(\beta_{\mathtt{UB}_k}-\breve{t}_{\mathtt{w}k})}{(N_\mathtt{t}-K)(\breve{t}_{\mathtt{w}j}+\tilde{t}_{\mathtt{w}j}\mathbf{e}^\T\check{\bm{\alpha}}^{[j]})^2}$.
As a result, 
the objective function in \eqref{equivalent_zf} is iteratively updated as follows
\begin{align}
\underset{\check{\boldsymbol{\alpha}}}
  {\text{maximize}} \sum_{k=1}^K  
\tilde{f}_k^{\mathtt{ZF}}(\bm{\alpha}; \bm{\alpha}^{[j]})
\label{zf_convex_approx_final},
\end{align}
where $\tilde{f}_k^{\mathtt{ZF}}(\bm{\alpha}; \bm{\alpha}^{[j]})$ is obtained from \eqref{obj_hris_recast_fig} by replacing $c_{kj}^\mathtt{ZF}(\check{\bm{\alpha}})$ with its affine approximation $\tilde{c}_{kj}^\mathtt{ZF}(\check{\bm{\alpha}}; \check{\bm{\alpha}}^{[j]})$ given in \eqref{c_kj_bound}.

\begin{algorithm}[t]
\caption{ AO-Based Algorithm to Solve  \eqref{eq:main_problem}}
\SetAlgoLined
\label{algo1}
\textbf{Initialization:} $m \leftarrow 0$, $\check{\bm{\alpha}}^{[0]}$, and $(\boldsymbol{\xi}^{[0]},\eta^{[0]})$\;
\Repeat{the objective value in \eqref{eq:main_problem} converges}{ 
$m \leftarrow m+1$\;
\textbf{Initialization:} $i \leftarrow 0$, $(\boldsymbol{\xi}^{[i]},\eta^{[i]}) \leftarrow (\boldsymbol{\xi}^{[m-1]},\eta^{[m-1]})$\; 
\Repeat{ convergence}
{ $i \leftarrow i+1$\;
 Update $(\boldsymbol{\xi}^{[i]},\eta^{[i]})$ by solving \eqref{power_allocation_convex_final} with CVX\;
}
Update $(\boldsymbol{\xi}^{[m]},\eta^{[m]})=(\boldsymbol{\xi}^{[i]},\eta^{[i]})$\;
\textbf{Initialization:} $j \leftarrow 0$, $\check{\bm{\alpha}}^{[j]} \leftarrow \check{\bm{\alpha}}^{[m-1]}$\; 
\Repeat{ convergence}
{$j \leftarrow j+1$\;
Update the auxiliary variables  $t_{\mathtt{u}k}^{[j]}$, $t_{\mathtt{v}k}^{[j]}$, and $t_{\mathtt{w}k}^{[j]}$  by \eqref{aux:t_u}, \eqref{aux:t_v}, and \eqref{aux:t_w}, respectively, $\forall k$\;
 Update $\check{\bm{\alpha}}^{[j]}$ by solving \eqref{final_mrt_convex} for MRT or \eqref{final_zf_convex} for ZF, with CVX\; 
}
 Update $\check{\bm{\alpha}}^{[m]}=\check{\bm{\alpha}}^{[j]}$\; 
}
\end{algorithm}

\subsubsection{Transformation of Constraints \eqref{CRLB_theta_constraint} and \eqref{CRLB_phi_constraint}}

Recall from \eqref{CRLB_psi} that $\mathtt{CRLB}^\mathtt{bf}_\psi(\check{\boldsymbol{\alpha}})$ is a convex function in $\check{\boldsymbol{\alpha}}$ when  ${\bm{\nu}^\mathtt{bf}}(\check{\boldsymbol{\alpha}})$ is affine in $\check{\boldsymbol{\alpha}}$. However, with the ZF precoder, the corresponding term ${\bm{\nu}^\mathtt{ZF}}(\check{\boldsymbol{\alpha}})$ is not affine. To tackle this, we employ the first-order Taylor approximation around $\check{\bm{\alpha}}^{[j]}$ and obtain
\begin{align}
 {{\bm{\nu}}^\mathtt{ZF}}(\check{\boldsymbol{\alpha}}) &\geq 
{{\bm{\nu}}}(\check{\boldsymbol{\alpha}}^{[j]})+ \frac{\tilde{t}_{\mathtt{w}k}\mathbf{e}^\T(\check{\bm{\alpha}}-\check{\bm{\alpha}}^{[j]})}{N_\mathtt{t}(N_\mathtt{t}-K)(\breve{t}_{\mathtt{w}k}+\tilde{t}_{\mathtt{w}k}\mathbf{e}^\T{\check{\bm{\alpha}}^{[j]}} )^2} \nonumber\\
 &\triangleq {\tilde{\bm{\nu}}^\mathtt{ZF}}(\check{\bm{\alpha}};\check{\boldsymbol{\alpha}}^{[j]}). 
 \label{lower_nu}
\end{align}
By substituting \eqref{lower_nu} into \eqref{CRLB_psi} yields a convex constraint that can be recast as a SOC constraints to enable an efficient solution. 
By following a similar approach as in \eqref{soc_constraint_theta}, constraints  \eqref{CRLB_theta_constraint} and \eqref{CRLB_phi_constraint} for the ZF precoder are iteratively updated as 
\begin{align}
\label{soc_constraint_zf_theta}
\bigg\|\bigg[2{
\left(\boldsymbol{\xi}^\T\tilde{\bm{\nu}}^{\mathtt{ZF}}(\check{\bm{\alpha}};\check{\boldsymbol{\alpha}}^{[j]})\varpi_{\theta\phi}
      + \eta N_{\mathtt{t}}\varphi_{\theta\phi}\right)};\tilde{\Delta}_\phi  -\bigg(\tilde{\Delta}_\theta\nonumber \\ \hspace{1cm}
     -\frac{1}{\lambda_\mathtt{0}\mathtt{CRLB}^\mathsf{th}_\theta}\bigg)\bigg]\bigg\|_2  \leq \tilde{\Delta}_\phi+\left(\tilde{\Delta}_\theta-\frac{1}{\lambda_\mathtt{0}\mathtt{CRLB}^\mathsf{th}_\theta}\right),
\end{align}
and 
\begin{align}
\label{soc_constraint_zf_phi}
\bigg\|\bigg[2{
\left(\boldsymbol{\xi}^\T\tilde{\bm{\nu}}^{\mathtt{ZF}}(\check{\bm{\alpha}};\check{\boldsymbol{\alpha}}^{[j]})\varpi_{\theta\phi}
      + \eta N_{\mathtt{t}}\varphi_{\theta\phi}\right)};\tilde{\Delta}_\theta  -\bigg(\tilde{\Delta}_\phi\nonumber \\ \hspace{1cm}
     -\frac{1}{\lambda_\mathtt{0}\mathtt{CRLB}^\mathsf{th}_\phi}\bigg)\bigg]\bigg\|_2  \leq \tilde{\Delta}_\theta+\left(\tilde{\Delta}_\phi-\frac{1}{\lambda_\mathtt{0}\mathtt{CRLB}^\mathsf{th}_\phi}\right),
\end{align}
respectively, where ${\lambda}_\mathtt{0}\triangleq \frac{{L|\beta_{\mathtt{s}}|^{2}N_{\mathtt{t}}N_{\mathtt{r}}\pi^{2}}}{{6\sigma_{\mathtt{s}}^{2}}}$, $\tilde{\Delta}_\theta \triangleq \boldsymbol{\xi}^\T\tilde{\bm{\nu}}^{\mathtt{ZF}}(\check{\bm{\alpha}};\check{\boldsymbol{\alpha}}^{[j]})\varpi_{\theta}
      + \eta N_{\mathtt{t}}\varphi_{\theta} $, and $\tilde{\Delta}_\phi \triangleq{\boldsymbol{\xi}^\T\tilde{\bm{\nu}}^{\mathtt{ZF}}(\check{\bm{\alpha}};\check{\boldsymbol{\alpha}}^{[j]})\varpi_{\phi}
      + \eta\,N_{\mathtt{t}}\varphi_{\phi}}$.
On the other hand, with the MRT precoder, we use the transformed SOC constraints \eqref{soc_constraint_theta} and \eqref{soc_constraint_phi} by setting  \(\mathtt{bf}=\mathtt{MRT}\).

\subsubsection{Transformation of Constraint \eqref{constraint_transmit power}}
{Note that \eqref{constraint_transmit power} for the ZF precoder} is non-convex due to the fractional term ${\bm{\nu}^\mathtt{ZF}}(\check{\boldsymbol{\alpha}})$. This can be tackled by substituting \eqref{lower_nu} into \eqref{constraint_transmit power}, leading to the following iterative update for the ZF precoder
\begin{align} 
\label{constraint_power_zf}\sum_{k=1}^KN_\mathtt{t}\xi_k{\tilde{\bm{\nu}}^\mathtt{ZF}}(\check{\bm{\alpha}};\check{\boldsymbol{\alpha}}^{[j]})+ N_\mathtt{t}\eta \leq P_\mathtt{t}.
\end{align}

Finally, problem \eqref{equivalent_mr} corresponding to the MRT precoder can be recast as the following convex problem 
\begin{align}
    \label{final_mrt_convex}\underset{\boldsymbol{\check{\bm{\alpha}}}}
  {\text{maximize}} \sum_{k=1}^K  
{f}^\mathtt{MRT}_k(\check{\bm{\alpha}}), \quad \text{s.t.}       \,\eqref{soc_constraint_theta},\eqref{soc_constraint_phi},\eqref{constraint_transmit power}-\eqref{hris_constraint1}.
\end{align}
 With the ZF precoder, the approximate convex problem associated with  \eqref{equivalent_zf}, solved at iteration $j$ 
 is given as
\begin{align}
    \label{final_zf_convex}\hspace{-0.1cm}\underset{\boldsymbol{\check{\bm{\alpha}}}}
  {\text{maximize}} \! \sum_{k=1}^K  \!
\tilde{f}^\mathtt{ZF}_k\!(\check{\bm{\alpha}};\check{\bm{\alpha}}^{[j]}), \, \text{s.t.}       \,\eqref{soc_constraint_zf_theta}-
\eqref{constraint_power_zf}, \eqref{hris_constaints}, \eqref{hris_constraint1}.
\end{align}

\subsection{Overall Algorithm and Complexity Analysis
}
{We summarize the SCA and FP based iterative algorithm
to solve \eqref{eq:main_problem} in Algorithm~\ref{algo1}.
The worst-case computational complexity of this algorithm is $\mathcal{O}(I_{\mathtt{out}}(I_{\boldsymbol{\xi},{\eta}}\sqrt{3}K^3+I_{\boldsymbol{\check{\alpha}}}N^3))$ \cite{nhan_isac_power_allocation}, \cite{shen_FP_part2}, where $I_{\mathtt{out}}$ is the number of iterations for the outer loop and $I_{\boldsymbol{\xi},{\eta}}$ and $I_{\boldsymbol{\check{\alpha}}}$ are the number of iterations for the inner loops corresponding to steps $5-8$ and steps $11-15$, respectively, in Algorithm~\ref{algo1}. 
Considering that $N\gg K$, the complexity
of the algorithm can be approximated as $\mathcal{O}(N^3)$, which is
required for the optimization of the HRIS coefficients.}

{The convergence of the first inner loop for updating $({\boldsymbol{\xi},{\eta}})$ follows from \cite{nhan_isac_power_allocation}. We therefore focus on the convergence of the second inner loop for updating  $\boldsymbol{\check{\alpha}}$.
Let $g_{\boldsymbol{\check{\alpha}}}(\boldsymbol{\check{\alpha}})$ and $f_{\boldsymbol{\check{\alpha}}}(\boldsymbol{\check{\alpha}},\mathbf t)$ denote the objective values of the original subproblem in \eqref{problem1} and the transformed subproblem in \eqref{lem:objecive_hris}, respectively, where $\mathbf t \triangleq \{\mathbf{t}_\mathtt{u}^\mathtt{bf},\mathbf{t}_\mathtt{v}^\mathtt{bf}, \mathbf{t}_\mathtt{w}\}$. From Lemma~\ref{lemma_FP}, we have $g_{\boldsymbol{\check{\alpha}}}(\boldsymbol{\check{\alpha}})\ge f_{\boldsymbol{\check{\alpha}}}(\boldsymbol{\check{\alpha}},\mathbf t)$ and $g_{\boldsymbol{\check{\alpha}}}(\boldsymbol{\check{\alpha}}^{[j-1]})=f_{\boldsymbol{\check{\alpha}}}(\boldsymbol{\check{\alpha}}^{[j-1]},\mathbf t^{[j]})$. For a fixed $\mathbf t^{[j]}$, the optimal value of $\boldsymbol{\check{\alpha}}^{[j]}$ is obtained by solving \eqref{final_mrt_convex} or \eqref{final_zf_convex}, yielding $f_{\boldsymbol{\check{\alpha}}}(\boldsymbol{\check{\alpha}}^{[j-1]},\mathbf t^{[j]}) \le f_{\boldsymbol{\check{\alpha}}}(\boldsymbol{\check{\alpha}}^{[j]},\mathbf t^{[j]})$. Based on this, we can write  $g_{\boldsymbol{\check{\alpha}}}(\boldsymbol{\check{\alpha}}^{[j-1]}) = f_{\boldsymbol{\check{\alpha}}}(\boldsymbol{\check{\alpha}}^{[j-1]},\mathbf t^{[j]}) \le f_{\boldsymbol{\check{\alpha}}}(\boldsymbol{\check{\alpha}}^{[j]},\mathbf t^{[j]}) \le g_{\boldsymbol{\check{\alpha}}}(\boldsymbol{\check{\alpha}}^{[j]})$. This implies that the inner loop generates a non-decreasing sequence of objective values that converge as $j$ increases.
Consequently, the outer loop of Algorithm~\ref{algo1}, which alternatively updates $(\boldsymbol{\xi},\eta)$ and $\boldsymbol{\check{\alpha}}$, yields a sequence of non-decreasing  objective values that converge to at least a local optimum  for a sufficiently large number of 
iterations \cite{r_liu_ris}
.}

\section{{Simulation Results}}
\label{sec:simulation_results}

In this section, we provide simulation results to verify the analytical derivations and proposed optimization algorithm. For comparison, we include the performance of the mMIMO ISAC system with PRIS and that without any RIS. 

\subsection{Simulation Setup}

We consider a scenario where
$K$ users are uniformly distributed within a circle of radius
$1000$m, with the BS located at the center $(0, 0)$. The RIS
is located at $(200, 100)$m. {Without loss of generality, we assume that the target is located in the non-boresight direction $(\theta,\phi)=\left(\frac{\pi}{12},\frac{\pi}{8}\right)$ \cite{uniyal_wcnc} and at a long-range sensing distance of $200$ m from the BS \cite{kpreeti_radar}.}
 The large-scale fading coefficients $(\beta_{\mathtt{BU}_k},\beta_\mathtt{BR},\beta_{\mathtt{RU}_k} )$ are modeled using the three-slope path loss model \cite{ngo}. In the simulations, we set $N_\mathtt{t}=225$, $N_\mathtt{r}=16$, $N=100$ \cite{nhan_isac_power_allocation}, \cite{sankar_ris}, $N_\mathsf{a}=4$, $\alpha_{\max}=5$dB \cite{nhan_hris_tvt}, $\mathtt{CRLB}_\psi^\mathtt{th}=-30$dB, $L=30$ \cite{nguyen_jsac_energy_efficiency}, $\sigma^2=\sigma_\mathtt{s}^2=\sigma_\mathtt{p}^2=1$, $\tau_\mathtt{c}=100$, $\tau_p=K$, and $p_\mathtt{p}=100$ \cite{ngo}. The target reflection coefficient is modeled as $\beta_\mathtt{s}=\frac{1+\mathrm{j}}{\beta_\mathtt{BT}\sqrt{2}}$, where $\beta_\mathtt{BT}$ is the round-trip path loss between the BS and target \cite{nhan_isac_power_allocation}. {The convergence tolerance of Algorithm~1 is set to $10^{-3}$ 
for both precoders.} All plots are generated with $1000$ small-scale and $50$ large-scale channel realizations.


\begin{figure}[!t]
	\centering
	\includegraphics[scale=0.46]{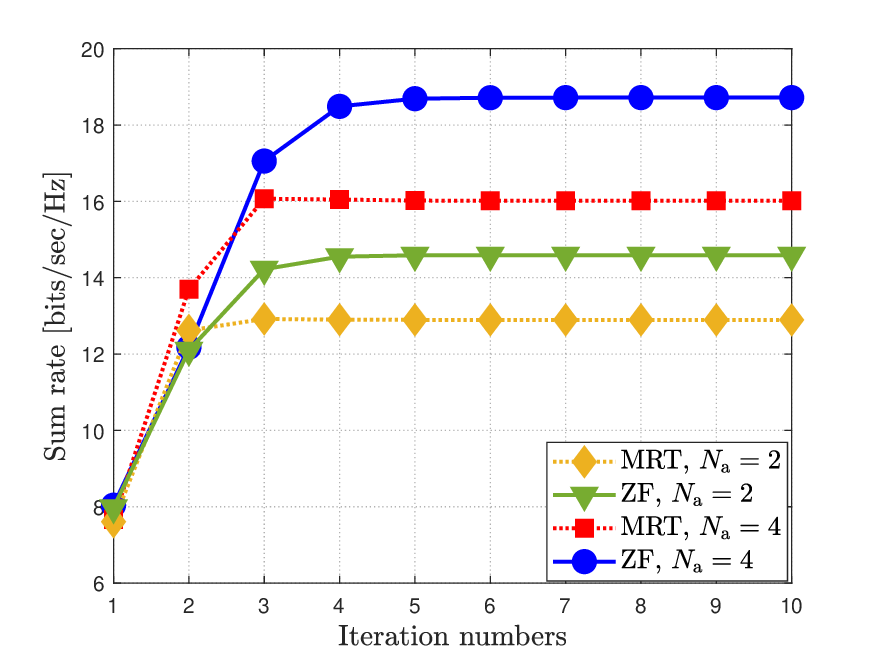}
	\caption{\small Convergence of Algorithm~\ref{algo1} {for 
    $N=100$, $N_\mathtt{t}=225$, $N_\mathtt{r}=16$, 
    $K=8$, $\text{SNR}=10$dB, and $\mathtt{CRLB}_\theta^\mathtt{th}=\mathtt{CRLB}_\phi^\mathtt{th}=-30$dB.}}
	\label{fig:convergence}
\end{figure}

\subsection{Convergence of Algorithm~\ref{algo1}}

We first show the convergence behavior of Algorithm~\ref{algo1} in Fig.~\ref{fig:convergence} for both MRT and ZF precoders. Specifically, we  initialize the algorithm by randomly generating a feasible $\check{\bm{\alpha}}^{[0]}$ that satisfies \eqref{hris_constaints} and \eqref{hris_constraint1}. We then set $\left(\xi_k^{[0]},\eta^{[0]}\right)=\left(\frac{P_\mathtt{t}}{2N_\mathtt{t}\sum_{k=1}^K\nu^\mathtt{bf}(\check{\bm{\alpha}}^{[0]})}, \frac{P_\mathtt{t}}{2N_\mathtt{t}}\right)$, ensuring the power constraint in \eqref{constraint_transmit power}.
 As can be seen, for the considered values of number of active elements $N_\mathtt{a}$ and the adopted precoders, the objective values monotonically increase and reach convergence in $3$--$6$ outer iterations.
 Furthermore, for a given $N_\mathtt{a}$,
the convergence is slightly better with MRT precoder compared to ZF precoder. This is because with the ZF precoder, the sum-rate maximization problem involves an additional step of computing a linear bound for the fractional term, as discussed in \eqref{zf_convex_approx_final}. {Nevertheless, the ZF precoder converges to a higher objective value than the MRT precoder.} Moreover, it is also observed that increasing $N_\mathtt{a}$ enhances the communication sum-rate with both MRT and ZF precoders.

\begin{figure*}[!t]
\small
    \centering
    \hspace{-5mm}
    \subfigure[Sum rate versus SNR.]
{\includegraphics[scale=0.44]{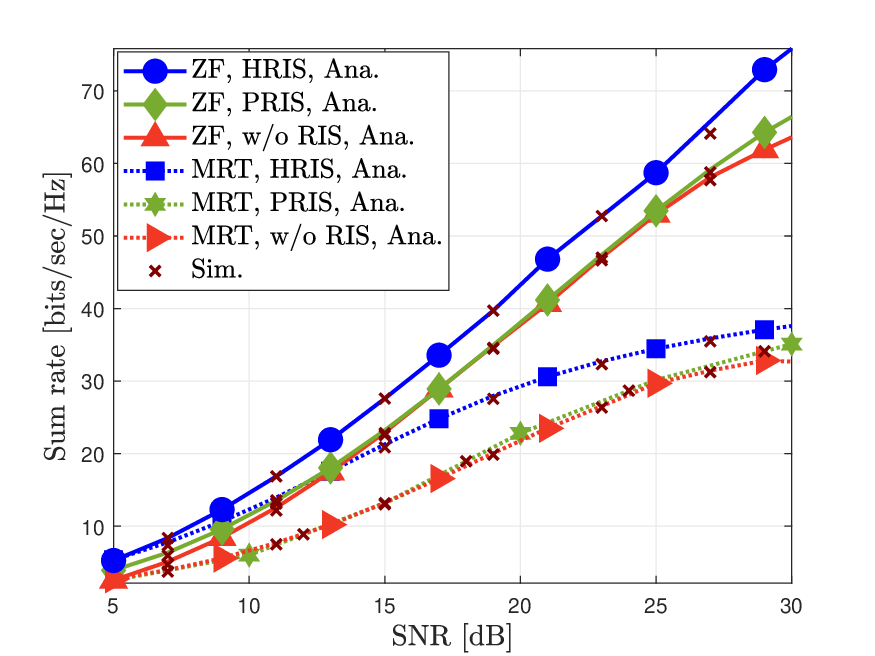}}
     \hspace{-8mm} 
    \subfigure[Sum rate versus $N_\mathtt{t}$]
    { \includegraphics[scale=0.44]{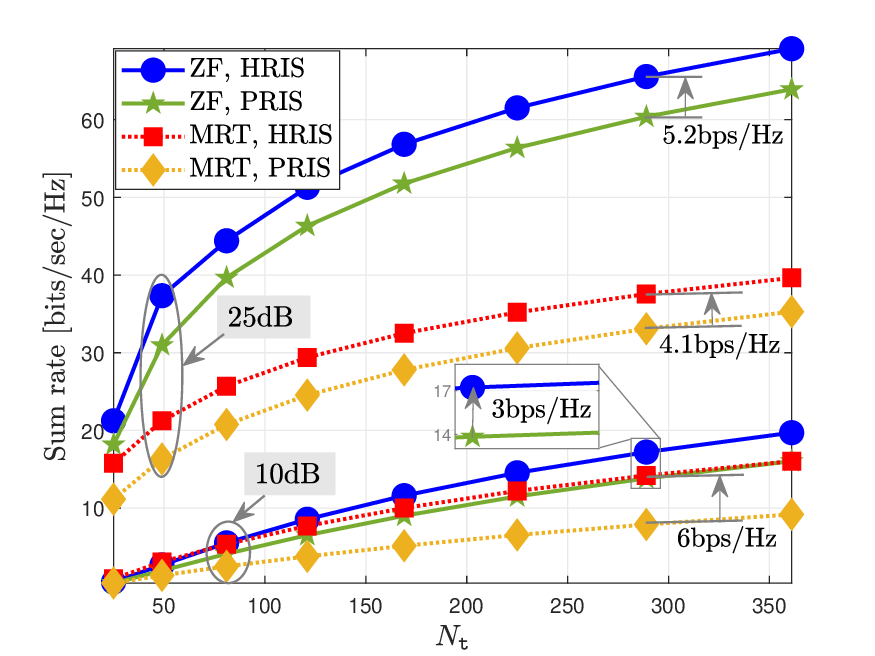}}
       \hspace{-9mm} 
    \subfigure[{Sum-rate versus CRLB threshold.}]
    { \includegraphics[scale=0.44]{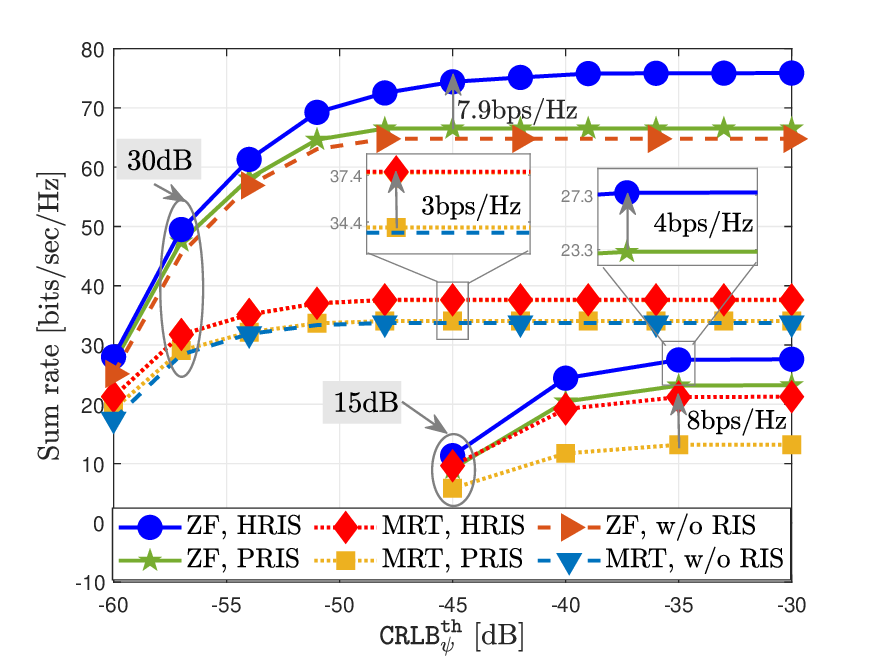}}
\caption{\small{Communication performance and ISAC trade-off with MRT and ZF precoders.}}
 \label{comms_performance_and_tradeoff}
\end{figure*}

\subsection{Communication Performance and Scaling Laws}

Fig.~\ref{comms_performance_and_tradeoff}(a) shows the achievable sum rate obtained
from the analytical expressions in \eqref{rate_final} and from the Monte Carlo simulations,
for various values of SNR.
First, it is observed that for all the considered cases, the analytical results align well with the simulation ones for the entire considered range of SNR, validating  Theorem~\ref{thm:sum_rate}. 
Subsequently, we observe that the PRIS yields only a marginal gain over the system without RIS mainly at high SNR. This is because PRIS provides only a weak cascaded-link gain, which becomes noticeable typically at high SNRs or for large values of $N$.
Furthermore, comparing the sum-rate of mMIMO ISAC system with and without HRIS, it is observed that the former clearly outperforms the latter across the entire considered range of SNR.

Fig.~\ref{comms_performance_and_tradeoff}(b) shows the sum-rate versus the number of transmit antennas $N_\mathtt{t}$ with ZF and MRT precoding schemes. It is seen that as $N_\mathtt{t}$ increases, the sum-rate for all the considered scenarios increases. This follows from \eqref{rate_final}, where the desired signal power grows linearly with $N_\mathtt{t}$, while the interference power remains independent of $N_\mathtt{t}$, yielding a higher SINR and sum-rate. Furthermore, the significance of the performance improvement of HRIS over PRIS depend on the precoder and SNR regime.  For example, at $\text{SNR}=10 $dB, HRIS provides substantial gains over PRIS with the MRT precoder, particularly at large $N_\mathtt{t}$. At $\text{SNR}=20 $dB, these gains are slightly higher with the ZF precoder than with the MRT precoder across all $N_\mathtt{t}$. This is consistent with \eqref{rate_final}, where with the MRT precoder, the HRIS enhances the desired signal power term \(u_k^\mathtt{MRT}(\check{\bm{\alpha}})\), which provides substantial gains in the power-limited regime. Conversely, with the ZF precoder, the HRIS reduces the inter-user interference term \(c_{kj}^\mathtt{ZF}(\check{\bm{\alpha}})\), leading to significant gains in the interference-limited regime. 
Moreover, the performance gap between the HRIS-aided and PRIS-aided systems increases with \(N_\mathtt{t}\) for both the precoders, revealing the synergistic benefit of large antenna arrays and HRIS.

In Fig.~\ref{comms_performance_and_tradeoff}(c), we show the communication sum-rate versus the CRLB threshold. It is seen that the 
sum-rate increases within the CRLB range of
$[-60, -45]$ dB for $\text{SNR}=30$dB and $[-45, -35]$ dB for $\text{SNR}=15$dB, beyond which it saturates. This is because with a high CRLB threshold, the sensing constraints \eqref{CRLB_theta_constraint} and \eqref{CRLB_phi_constraint} can be  easily satisfied, and more power can be allocated for communication. It is clear that the ZF-HRIS achieves the best tradeoff between the communication sum-rate and sensing CRLB. Moreover, as expected, MRT-HRIS provides larger performance gains compared to MRT-PRIS at moderate SNR, i.e., $15$dB, while ZF-HRIS provides larger performance gains over ZF-PRIS at a comparatively higher SNR, i.e., $30$dB.
{It is observed that, with ZF precoding at an SNR of $30$~dB,
achieving a sum rate of $65$~bits/s/Hz requires a CRLB threshold of at
least $-48$~dB without an RIS, whereas the HRIS-assisted system achieves
the same sum rate at $-53$~dB. Thus, the HRIS enables the same
communications performance under a stricter sensing requirement.}

In Fig.~\ref{fig:sr_scaling_laws} we investigate the power-scaling laws of the considered system with respect to the total number of HRIS elements, i.e., $N$. We observe that when the transmit power is scaled as $P_\mathtt{t}=e_\mathtt{t}/N^\varepsilon$, for $\varepsilon\leq1$, the system achieves a non-zero sum-rate despite the pilot contamination, while ensuring $\mathtt{CRLB}_\psi=-5$dB. For moderate values of $N$, the sum-rate increases rapidly. This is because, for moderate values of $N$, the reduction in transmit power $P_\mathtt{t}=e_\mathtt{t}/N^\varepsilon$ is relatively small and the gain in sum-rate from increasing $N$ is larger than the loss from the reduction in transmit power.
 As $N$ becomes very large, the scaling-law depends critically on the power scaling factor $\varepsilon$. For $\varepsilon=1$, the sum-rate saturates at large $N$ because the gain from increasing $N$ is balanced by the $1/N$ decrease in transmit power. In contrast, for $\varepsilon=0.5$, the reduction in $P_\mathtt{t}$ is less severe, which allows the sum-rate to increase gradually with $N$. These results are consistent with the theoretical findings in \eqref{rate_asymptotic}.
{ Furthermore, in Fig.~\ref{fig:sr_scaling_laws}(a), it is seen that for equal power allocation,} the performance with the ZF and MRT precoders is nearly the same. In contrast, Fig.~\ref{fig:sr_scaling_laws}(b) shows that the ZF precoder consistently outperforms the MRT precoder for all considered values of $N$. Moreover, the sum-rates achieved with the proposed algorithm are substantially higher than those with equal power allocation.
{We further investigate the scaling laws of the ARIS-aided ISAC system under an active-power budget equal to that of the HRIS. As shown in Fig.~\ref{fig:sr_scaling_laws}(b), for $\varepsilon=0.5$, the sum rate of the ARIS-aided system decreases as $N$ increases. In contrast, the HRIS-aided system continues to benefit from a larger number of HRIS elements over the entire considered range of $N$. This behavior is consistent with the discussion in Remark~\ref{rem:sum_rate_scaling}.}

\begin{figure}[!t]
	\centering
 \vspace{-0.3cm}
	\includegraphics[scale=0.45]
{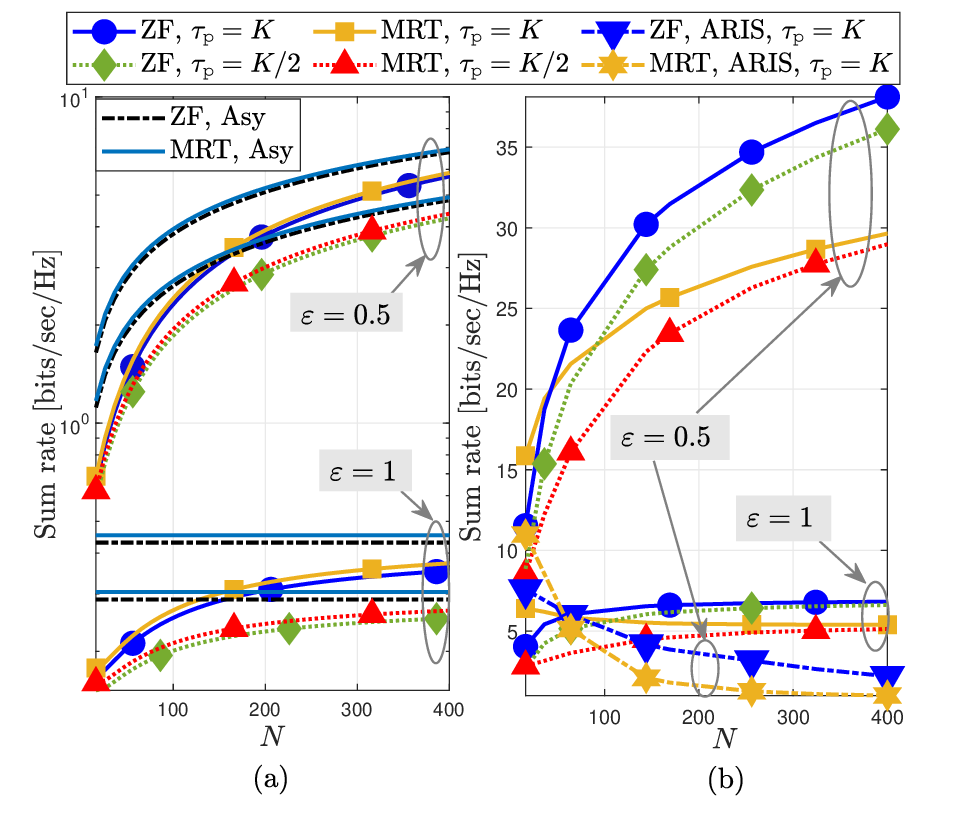}
	\caption{\small {{Power scaling laws with} (a) equal power allocation scheme {(b) proposed scheme for}} $P_\mathtt{t}=\frac{25\text{dB}}{N^\varepsilon}$ and $\mathtt{CRLB}^\mathtt{th}_\psi=-5$dB.}
	\label{fig:sr_scaling_laws}
\end{figure}

\begin{figure}[!t]
	\centering
 \vspace{-0.5cm}
	\includegraphics[scale=0.46]{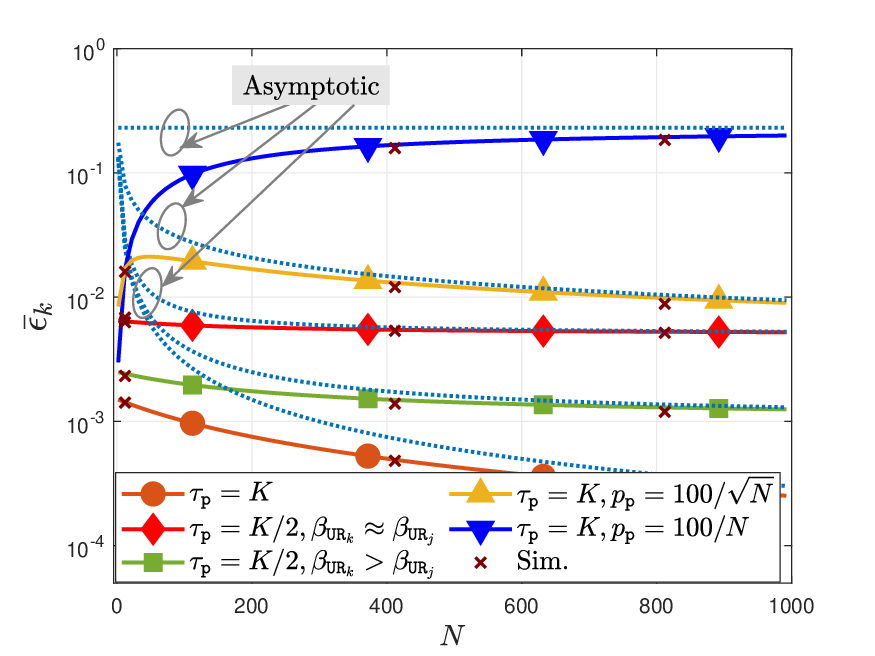}
	\caption{\small NMSE of user $k$ versus the number of HRIS elements $N$.}
	\label{fig:nmse}
\end{figure}

 \begin{figure*}[!t]
\small
    \centering
    \hspace{-5mm}
    \subfigure[Sum-rate versus $N$.]
{\includegraphics[scale=0.44]{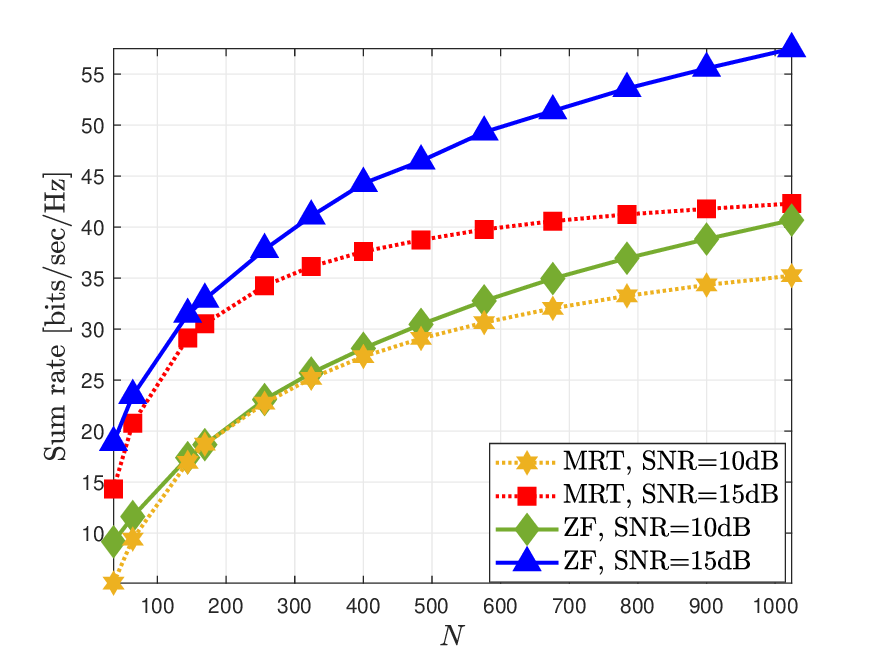}}
     \hspace{-8mm} 
    \subfigure[{CRLB versus $N$ for SNR $=15$ dB.}]
    { \includegraphics[scale=0.44]{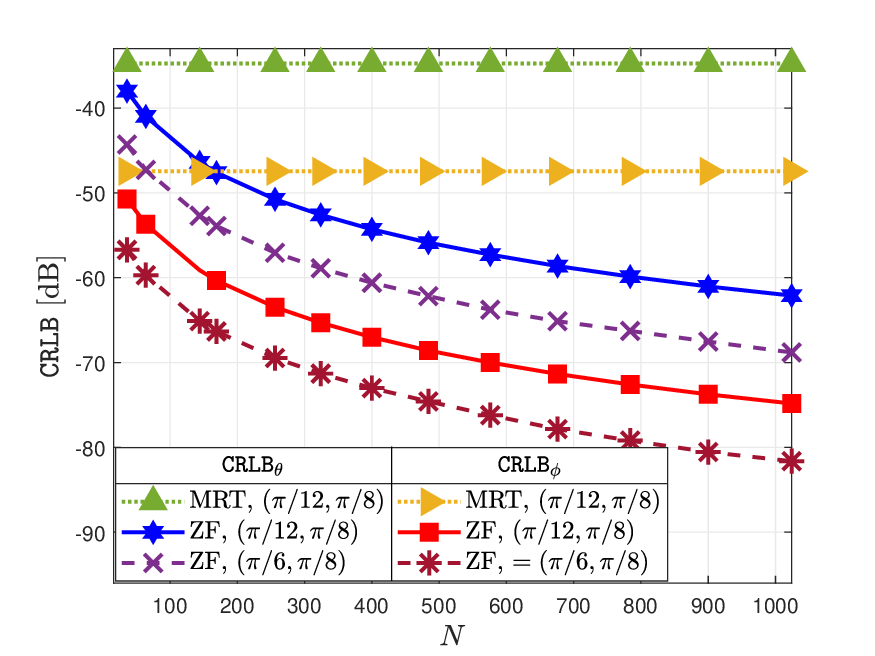}}
       \hspace{-9mm} 
    \subfigure[Power allocation versus $N$.]
    { \includegraphics[scale=0.44]{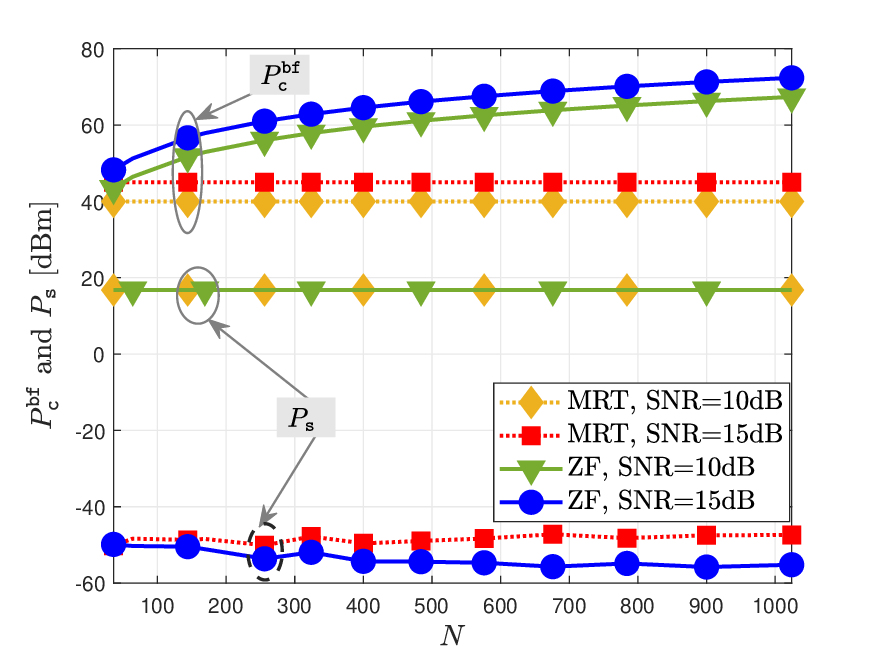}}
\caption{\small{Impact of total number of HRIS elements $(N)$ on the communication and sensing performances.}}
    \label{fig:performance} 
\end{figure*}

\subsection{NMSE of MMSE Estimate and Power Scaling Laws}

In Fig.~\ref{fig:nmse}, we show the
NMSE of the channel estimate for 
user-$k$ versus the number of HRIS 
elements $N$.
It is seen that the analytical results align well with those obtained
via Monte Carlo simulations, validating the exact and asymptotic analysis provided in \eqref{eq:nmse} and Remark~\ref{rem:nmse_scaling}, respectively.
Furthermore, to illustrate the performance loss due to  pilot contamination,  we set $\tau_\mathtt{p}={K}/{2}$ with $\beta_{\mathtt{UR}_k} > \beta_{\mathtt{UR}_j}$ and $\beta_{\mathtt{UR}_k}  \approx \beta_{\mathtt{UR}_j}$ for weak and strong
contamination, respectively.
As can be seen, the impact of pilot contamination on the channel estimation is more pronounced when the strength of the contaminating channel is comparable to that of the channel being estimated.  
It is also observed that  $\bar{\epsilon}_k$ converges to a non-zero constant (error floor) for large $N$ when $\tau_\mathtt{p}<K$, whereas for $\tau_\mathtt{p}=K$, no such error floor occurs. This is consistent with the findings in Remark~\ref{rem:nmse_scaling}, which highlights the negative impact
of pilot contamination on the channel estimation accuracy. Finally, we observe that, $\bar{\epsilon}_k$ converges to a constant  
when the pilot power scales as 
$p_\mathtt{p}=\frac{100}{N^\varepsilon}$ for $\varepsilon \in \{0.5,1\}$. 
This
shows that the HRIS can maintain a bounded NMSE even at very low pilot power.

\begin{figure}[!t]
	\centering
	\includegraphics[scale=0.45]{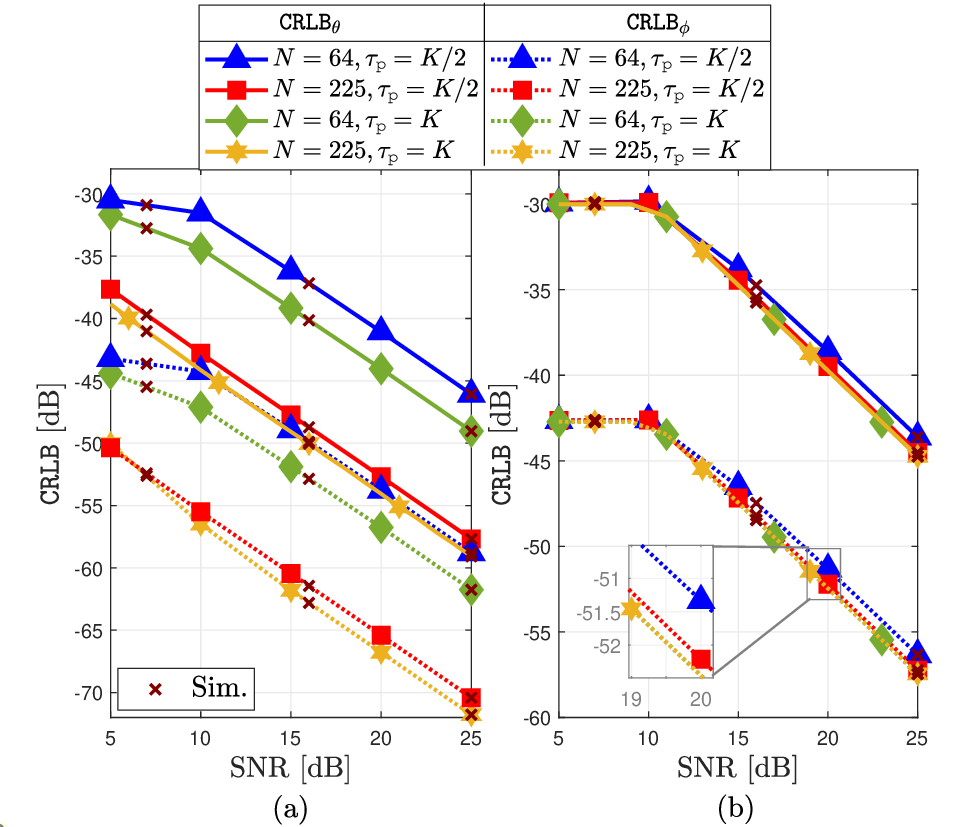}
	\caption{\small CRLB versus SNR with (a) ZF and (b) MRT precoders.}
	\label{fig:crlb_snr}
\end{figure}

\subsection{Impact of Total Number of HRIS Elements ($N$)}
In Fig.~\ref{fig:performance}, we show the impact $N$ on the sum-rate and the CRLB  for $\text{SNR} = \{10, 15\}$~dB. From Fig.~\ref{fig:performance}(a), it can be seen that the sum-rate initially increases rapidly with $N$ and then tends to saturate as $N$ becomes large, which aligns with the analysis in Remark~\ref{rem:sum_rate_scaling}. This behavior differs from that of a purely active RIS-assisted system, 
where, under a fixed power budget, the sum-rate eventually decreases as 
$N$ grows large.
Furthermore, the ZF precoder outperforms the MRT precoder, and the performance gap between them gradually increases with $N$, especially for $\text{SNR} = 15$dB.

In Fig.~\ref{fig:performance}(b), we show the impact of $N$ on the CRLB. It is observed that, with the ZF precoder, the CRLBs for $\theta$ and $\phi$ decrease as $N$ increases. In contrast, with the MRT precoder, the CRLB remains constant with $N$. This is consistent with the power allocated for communication and sensing, i.e., $P^\mathtt{bf}_\mathtt{c}$ and $P_\mathtt{s}$, respectively, shown in Fig.~\ref{fig:performance}(c). More specifically, the reduction in CRLB for the ZF case is reasonable because, as $N$ increases, $P^\mathtt{bf}_\mathtt{c}$
increases. Simultaneously, $P_\mathtt{c}^\mathtt{bf}$ also contributes to the sensing performance, thereby reducing the CRLB as observed in \eqref{CRLB_psi} and  Fig.~\ref{fig:performance}(b). Conversely, for MRT, the CRLB remains constant since both $P^\mathtt{bf}_\mathtt{c}$ and $P_\mathtt{s}$
remain constant 
with $N$. 
Fig.~\ref{fig:performance}(c) further shows that at $\text{SNR} =10$dB, $P_\mathtt{s}$ is significantly high to satisfy the CRLB constraint. As the SNR increases to $15$dB, $P_\mathtt{c}^\mathtt{bf}$ increases to maximize the sum-rate, which also reduces the CRLB. 
Consequently, higher SNR yields a significant reduction in $P_\mathtt{s}$. {Furthermore, Fig.~\ref{fig:performance}(b) presents the ZF-based CRLB for two target-angle settings, $(\theta,\phi)=(\pi/12,\pi/8)$ and $(\theta,\phi)=(\pi/6,\pi/8)$. Although the two settings yield different CRLB values, the CRLB decreases monotonically with $N$ in both cases. This demonstrates that the CRLB trend with respect to $N$ is not specific to a particular target-angle setting.}

Finally, in Figs.~\ref{fig:crlb_snr}(a) and \ref{fig:crlb_snr}(b) we investigate the impact of pilot contamination on the CRLB across various SNR values, with the ZF and MRT precoders, respectively. The curves for the CRLB are obtained using the expression in \eqref{CRLB_psi} and from the Monte-Carlo simulations. Across all considered scenarios, the analytical results closely align with the simulation ones, validating 
Theorem~\ref{thm:CRLB}. As can be seen, pilot contamination in the communication channel also degrades the sensing performance, leading to higher CRLB values.
Nevertheless, the impact of pilot contamination becomes less severe with a larger $N$. 
Fig.~\ref{fig:crlb_snr}(a) shows that with the ZF precoder, the impact of pilot contamination on the CRLB reduces from \(8\%\) to \(4\%\) as \(N\) increases from $64 $ to $225$. Likewise, Fig.~\ref{fig:crlb_snr}(b) shows that with the MRT precoder, this impact reduces from \(3.7\%\) to \(1.8\%\) for the same range of $N$. This highlights the capability of the HRIS to mitigate the negative impact of pilot contamination, validating our discussion in Remark~\ref{rem:nmse_scaling}. 
 \subsection{{EE Comparison Between the ARIS and HRIS}}

\begin{figure}[!t]
	\centering
	\includegraphics[scale=0.46]{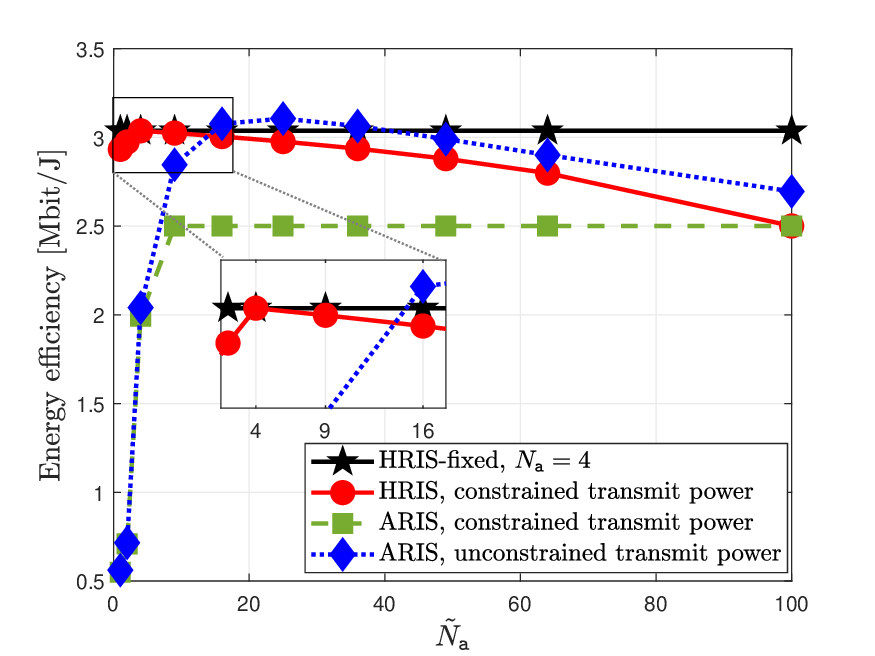}
	\caption{{\small{EE comparison between the ARIS-aided and HRIS-aided systems.}}}
	\label{fig:ee}
\end{figure}
{Fig.~\ref{fig:ee} illustrates the energy efficiency (EE) of the ARIS- and HRIS-aided systems versus the number of active elements, $\tilde{N_{\mathtt{a}}}$. We compare the following RIS architectures:
\begin{itemize}
    \item \emph{HRIS-fixed:} An HRIS with $N=100$ elements, consisting of a fixed number of $N_{\mathtt{a}}=4$ active elements and $N-N_{\mathtt{a}}=96$ passive elements.
    \item \emph{HRIS (constrained transmit power):} An HRIS with $\tilde{N}_{\mathtt{a}}$ active elements and $N-\tilde{N}_{\mathtt{a}}$ passive elements, where the total transmit power of the active elements is constrained to be the same as that of the \emph{HRIS-fixed} architecture. Unlike \emph{HRIS-fixed}, the numbers of active and passive elements vary with $\tilde{N}_{\mathtt{a}}$. In particular, as $\tilde{N}_{\mathtt{a}}$ increases, the number of passive elements decreases accordingly. When $\tilde{N}_{\mathtt{a}}=N_{\mathtt{a}}$, this architecture reduces to \emph{HRIS-fixed}.
    \item \emph{ARIS (unconstrained transmit power):} An ARIS with $N=\tilde{N}_{\mathtt{a}}$ elements, all of which are active, with no transmit power constraint imposed.
    \item \emph{ARIS (constrained transmit power):} The same architecture as the previous case, except that its total transmit power is constrained to be equal to that of the \emph{HRIS-fixed} architecture for a fair comparison.
\end{itemize}} 
{The EE is computed as
$\mathrm{EE} = \dfrac{B_{\mathtt{W}}\sum_{k=1}^{K}
\mathcal{R}_k^{\mathtt{bf}}(\boldsymbol{\xi},\eta,\check{\boldsymbol{\alpha}})}
{P_{\mathtt{tot}}^{\mathtt{sys}}}$,
where $B_{\mathtt{W}}=20$~MHz is the system bandwidth and
$P_{\mathtt{tot}}^{\mathtt{sys}}$ is the total consumed power is given by
$
P_{\mathtt{tot}}^{\mathtt{sys}}
=
\frac{P_{\mathtt{tot}}^{\mathtt{bf}}}{\tau_{\mathtt{BS}}}
+N_{\mathtt{t}}P_{\mathtt{BS}}
+P_{\mathtt{LO}}
+P_{\mathtt{RIS}}$,
where, $\tau_{\mathtt{BS}}$ is the BS
power-amplifier efficiency, $P_{\mathtt{BS}}$ is the power required to
operate the circuit components associated with each BS antenna, and
$P_{\mathtt{LO}}$ is the power consumed by the local oscillator. The power consumption of the RIS, unified across all considered RIS architectures via $\tilde{N}_{\mathtt{a}}$, is given by \cite[(9), (14)]{long2021active}
\begin{align}\label{eq_Pris}
    P_{\mathtt{RIS}}
=
\widetilde{N}_{\mathtt{a}}
\left(P_{\mathtt{c}}+P_{\mathtt{DC}}\right)
+
\left(N-\widetilde{N}_{\mathtt{a}}\right)P_{\mathtt{c}}
+
\frac{P_{\mathtt{a}}}{\upsilon},
\end{align}}
{where $P_{\mathtt{c}}$ is the
switch-and-control circuit power consumed by each RIS element,
$P_{\mathtt{DC}}$ is the additional DC-bias power consumed by each active
element, and $\upsilon$ denotes the active-load amplifier
efficiency. Here, $P_{\mathtt{a}}$
denotes the total output signal power
of the active elements, and is given as
\begin{align}
\nonumber
P_{\mathtt{a}} &= \mathbb{E}\{\mathtt{tr}(\boldsymbol{\Theta}_{\mathtt{a}}(\mathbf{H}_{\mathtt{BR}}
\mathbf{W}\mathbf{W}^{\H}\mathbf{H}_{\mathtt{BR}}^{\H}
+ \sigma_{\mathtt{a}}^2\mathbf{I}_N)\boldsymbol{\Theta}_{\mathtt{a}}^{\H})\}  \\ \nonumber
&= \sum_{n\in\mathcal{I}_{\mathtt{a}}} |\alpha_n|^2
\left(\mathbb{E}\{\|\mathbf{h}_n\|^2\} + \sigma_{\mathtt{a}}^2\right),
\end{align}}} 
{where $\mathbf{h}_n^\H \in \mathbb{C}^{1 \times K}$ denotes the $n$-th row of the matrix $\mathbf{H}_{\mathtt{BR}}\mathbf{W}$.
We set $\tau_{\mathtt{BS}}=0.5$ \cite{nhan_hris_tvt}, $P_{\mathtt{BS}}=1$~W,
and $P_{\mathtt{LO}}=2$~W \cite{bjornson2015optimal}, and $P_{\mathtt{c}}=-10$~dBm,
$P_{\mathtt{DC}}=-5$~dBm, and $\upsilon=0.8$ \cite{long2021active}.} 


{It can be seen from Fig.\ \ref{fig:ee} that the unconstrained ARIS slightly outperforms the
HRIS-fixed in terms of EE for $\tilde{N}_{\mathtt a}\in[16,36]$. However, its EE
decreases for larger $\tilde{N}_{\mathtt a}$ because of the substantial increase
in power consumption based on \eqref{eq_Pris}. Under the constrained transmit-power setting, the HRIS consistently outperforms the
ARIS in terms of EE over the entire considered range of $\tilde{N}_\mathtt{a}$. For small $\tilde{N}_{\mathtt{a}}$, both ARIS schemes also achieve lower EE  than the HRIS because they do not benefit from the passive beamforming gain provided by passive RIS elements. More importantly, the HRIS-fixed achieves
the highest EE, outperforming both the ARIS and the HRIS
employing a larger number of active elements. The results demonstrate that the HRIS architecture with a few active elements is sufficient to achieve satisfactory EE.}

\begin{figure}[!t]
	\centering
	\includegraphics[scale=0.46]{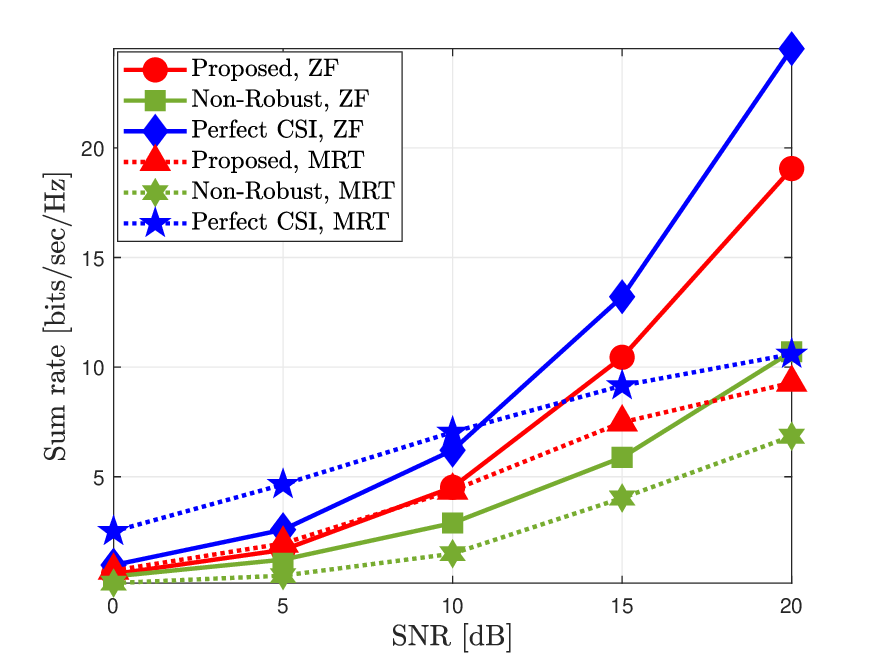}
	\caption{\small{{Comparison of the proposed imperfect CSI-based design with the perfect CSI benchmarks.}}}
	\label{fig:pcsi}
\end{figure}
\subsection{{Comparison Between Imperfect- and Perfect-CSI Designs}}
{In Fig.~\ref{fig:pcsi}, we compare the sum-rate performance of the proposed design with that of the perfect CSI-based benchmarks.
Specifically, we consider the following two cases:
(i) \emph{Perfect CSI}: In this case the power allocation and HRIS coefficients are
optimized and evaluated using the true channels; and
(ii) \emph{Non-Robust}: In this case the perfect CSI-based solution is evaluated over the actual imperfect channels.
It is observed that the perfect CSI case achieves the highest sum rate for
both MRT and ZF precoders. However, this case is not practically attainable because it requires perfect CSI. Moreover, the non-robust case shows that evaluating the perfect CSI-based solution over the actual imperfect channels results in a substantial loss in sum rate. For example, at $\mathrm{SNR}=15$~dB, the sum rate decreases by approximately $56\%$ for both the ZF and MRT precoders relative to the perfect CSI case. In contrast, the proposed design incurs sum-rate losses of around $21\%$ and $17\%$ for the ZF and MRT precoders, respectively, relative to the perfect CSI case. The proposed design also clearly outperforms the non-robust
benchmark, showing the benefit of accounting for channel-estimation
errors in the design.}

\section{Conclusion}
\label{sec:conclusion}

We have investigated the performance of an HRIS-assisted monostatic mMIMO ISAC system, employing both MRT and ZF precoders. We have 
derived the closed-form expressions for the achievable sum-rate and CRLB to assess the performances of the communication and sensing subsystems, respectively. In addition, we have discussed the power scaling laws with respect to the total number of HRIS elements to provide key insights into HRIS design. Then, we have presented an AO-based algorithm to solve the joint optimization problem of power allocation and HRIS coefficients to maximize the sum-rate, while satisfying the CRLB constraint. 
The analytical results and the proposed algorithm have been numerically justified via extensive simulations, revealing the synergistic benefits of employing large antenna arrays
and HRIS to enhance both communication and sensing performances in the considered HRIS-aided mMIMO ISAC system.
{The joint optimization of the transceiver and HRIS beamforming using a sensing-centric CRLB-based objective, along with the design of low-complexity algorithms, constitutes an important direction for future research. Furthermore, extending the proposed framework to simultaneously transmitting and reflecting (STAR)-RIS-assisted systems and exploring deep reinforcement learning-based methods constitute other important directions for future work.}

\appendices

\section{proof of lemma~\ref{lemma:transmit_covariance}} \label{sec:transmit_covaraince}
\renewcommand{\theequation}{A.\arabic{equation}}
\renewcommand{\theHequation}{A.\arabic{equation}}
\setcounter{equation}{0}

The covariance matrix $\mathbf{R}_\mathbf{x}^\mathtt{bf}$ is computed as $\mathbf{R}_\mathbf{x}^\mathtt{bf}=\mathbb{E\{\mathbf{W}^\mathtt{bf}(\mathbf{W}^\mathtt{bf})^\H\}}$. From \eqref{dual_bf}, for the MRT case, we have
$\mathbf{R}_\mathbf{x}^\mathtt{MRT}\
\overset{\text{(a)}}{=}\mathbb{E}\{
\hat{\mathbf G}\boldsymbol{\Xi}^2\hat{\mathbf G}^\H
+ \eta {\mathbf u \mathbf u^\H}\}  =\sum_{k=1}^{K} \xi_k \mathbb{E}\left\{\hat{\mathbf g}_k \hat{\mathbf g}_k^\H\right\}
+\eta{\mathbf u \mathbf u^\H} \overset{\text{(b)}}{=}\sum_{k=1}^{K} \xi_k \hat{\varrho}_k\mathbf{I}_{N_\mathtt{t}}
+ \eta{\mathbf u \mathbf u^{\H}}$.
where $\eta=\|\boldsymbol{\rho}\|^2$. Besides, the equality $(\text{a})$ follows that $\mathbb{E}\{\hat{\mathbf{g}}_k\}=\mathbf{0}$ and equality $(\text{b})$ follows from \eqref{channel_estimate_error}. 
Similarly, for ZF, we have
$\mathbf{R}_\mathbf{x}^\mathtt{ZF}
= \mathbb{E}\big\{
  \hat{\mathbf{G}}
  \big(\hat{\mathbf{G}}^{\mathrm{H}}\hat{\mathbf{G}}\big)^{-1}
  \mathbf{\Xi}^2
  \big(\hat{\mathbf{G}}^{\mathrm{H}}\hat{\mathbf{G}}\big)^{-1}
  \hat{\mathbf{G}}^{\mathrm{H}}
\big\}
+ \eta {\mathbf{u}\mathbf{u}^{{\H}}} \overset{\text{(c)}}{=} \mathbb E\big\{
\mathbf Z (\mathbf Z^{\H}\mathbf Z)^{-1}
\bold{D}^{-\frac{1}{2}}
\boldsymbol{\Xi}^2
\bold{D}^{-\frac{1}{2}}
(\mathbf Z^{\H}\mathbf Z)^{-1}\mathbf Z^{H}
\big\} \overset{\text{(d)}}{=}\frac{1}{N_\mathtt{t}(N_\mathtt{t}-K)}\sum_{k=1}^K {\xi_k}{\hat{\varrho}_k^{-1}}\mathbf{I}_{N_\mathtt{t}}+\eta{\mathbf{u}\mathbf{u}^{{\H}}}$.
Here, the equality 
$\text{(c)}$ follows that $\hat{\mathbf{G}}=\mathbf{D}^{\frac{1}{2}}\mathbf{Z}$, where $\mathbf{D}=\mathtt{diag}(\hat{\varrho}_1,\cdots,\hat{\varrho}_K)$ and $\mathbf{Z}=[\mathbf{z}_1,\cdots, \mathbf{z}_K]$ with $\mathbf{z}_k\sim\mathcal{CN}(\mathbf{0},\mathbf{I}_{N_\mathtt{t}})$. The equality $\text{(d)}$ is achieved using 
\cite[Lemma~1]{liao}.  {Finally, taking the trace of $\mathbf{R}_{\mathbf{x}}^{\mathtt{bf}}$ yields the total expression for the transmit power in \eqref{eq:total_transmit_power}.}


\section{proof of Theorem~\ref{thm:sum_rate}}\label{sec:sum_rate}
\renewcommand{\theequation}{B.\arabic{equation}}
\renewcommand{\theHequation}{B.\arabic{equation}}
\setcounter{equation}{0}
To derive the closed-form achievable rate, we compute $\mathbb{E}\{\bold{g}^\H_k\bold{w}^\mathtt{bf}_k\}$, $\mathbb{E}\{|\bold{g}^\H_k\bold{w}^\mathtt{bf}_k|^2\}$, and $\mathbb{E}\{|\bold{g}^\H_k\bold{w}^\mathtt{bf}_j|^2\}$ in \eqref{eq_SINR} as follows.

\subsubsection{MRT beamformer} In this case we have $\mathbf{f}^\mathtt{MRT}_k=\hat{\bold{g}}_k$.

$\bullet$ Compute $\mathbb{E}\{\bold{g}^\H_k\bold{w}^\mathtt{bf}_k\}$:
We rewrite
   $\mathbb{E}\{\bold{g}^\H_k\bold{w}^\mathtt{MRT}_k\}= \mathbb{E}\{\hat{\bold{g}}^\H_k\bold{w}^\mathtt{MRT}_k\} +\mathbb{E}\{\tilde{\bold{g}}^\H_k\bold{w}^\mathtt{MRT}_k\} =  \mathbb{E}\{\hat{\bold{g}}^\H_k\bold{w}^\mathtt{MRT}_k\} $, where the last equality holds due to independence of $\hat{\bold{g}}_k$ and  $\tilde{\bold{g}}_k$, and due to the fact that $ \mathbb{E}\{\tilde{\bold{g}}_k\}=0$. Furthermore, given that $\bold{w}^\mathtt{MRT}_k=\sqrt{\xi_k}\hat{\bold{g}}_k +\sqrt{\rho_k}\bold{a}$, we write $\mathbb{E}\{\hat{\bold{g}}^\H_k\bold{w}^\mathtt{MRT}_k\}=\mathbb{E}\{\hat{\bold{g}}^\H_k(\sqrt{\xi_k}\hat{\bold{g}}_k +\sqrt{\rho_k}\bold{a})\}$. 
Then, by using
   Theorem \ref{thm:mmse_estimate}, we obtain  
\begin{align}
\mathbb{E}\{\bold{g}^\H_k\bold{w}^\mathtt{MRT}_k\}=\sqrt{\xi_k}\mathsf{tr}\left(\hat{\bold{R}}_k\right)=\sqrt{\xi_k}\hat{\varrho}_kN_\mathtt{t}.
\label{compute_d_mrt}
\end{align}

$\bullet$ Compute $\mathbb{E}\{|\bold{g}^\H_k\bold{w}^\mathtt{MRT}_k|^2\}$: Here, we have
\begin{align}
\mathbb{E}\{|{\bold{g}}_k\bold{w}^\mathtt{MRT}_k|^2\} &= \mathbb{E}\{|(\hat{\bold{g}}_k+\tilde{\bold{g}}_k)^\H)(\sqrt{\xi_k}\hat{\bold{g}}_k +\sqrt{\rho_k}\bold{a})|^2\}
\nonumber \\
&={\xi_k}\mathbb{E}\{\|\hat{\bold{g}}_k\|^4\} +{\rho_k}\mathbb{E}\{|\hat{\bold{g}}_k^\H\bold{a}|^2\} \nonumber  \\ 
&\hspace{0.5cm}+\mathbb{E}\{|\tilde{\bold{g}}_k^\H(\sqrt{\xi_k}\hat{\bold{g}}_k +\sqrt{\rho_k}\bold{a})|^2\}\nonumber\\
&=\xi_k |\mathsf{tr}(\hat{\bold{R}}_k)|^2+ \xi_k\mathsf{tr}\left(\hat{\bold{R}}_k^2\right)+\rho_k\left(\mathbf{a}^{\H}\hat{\bold{R}}_k\mathbf{a}\right)\nonumber \\ &\hspace{0.5cm}+\xi_k \mathsf{tr}(\tilde{\mathbf{R}}_k \hat{\mathbf{R}}_k)+\rho_k\left(\mathbf{a}^{\H}\tilde{\bold{R}}_k\mathbf{a}\right),
\label{compute_bu_mrt}
\end{align}
where the last equality follows from \cite[Appendix B]{uniyal_wcnc}.
By substituting ${\mathbf{R}}_k={\varrho}_k\mathbf{I}_{N_\mathtt{t}}$, $\tilde{\mathbf{R}}_k=\tilde{\varrho}_k\mathbf{I}_{N_\mathtt{t}}$, and $\hat{\mathbf{R}}_k=\hat{\varrho}_k\mathbf{I}_{N_\mathtt{t}}$  from Theorem~\ref{thm:mmse_estimate} into  \eqref{compute_bu_mrt}, we get
\begin{align}
\label{compute_bu1_mrt}
    \mathbb{E}\{|{\bold{g}}_k\bold{w}^\mathtt{MRT}_k|^2\}
=\xi_k \hat{\varrho}_k^2N_\mathtt{t}^2+ \xi_k {{\varrho}}_k \hat{{\varrho}}_kN_\mathtt{t}+ \rho_k {{\varrho}}_k N_\mathtt{t}.
\end{align}

$\bullet$ Compute $\mathbb{E}\{|\bold{g}^\H_k\bold{w}^\mathtt{MRT}_j|^2\}$:  Similar to \eqref{compute_bu_mrt}, we can write
\begin{align}
    \mathbb{E}\{|\bold{g}^\H_k\bold{w}^\mathtt{MRT}_j|^2\}&={\xi_j}\mathbb{E}\{|\hat{\bold{g}}^\H_k\hat{\bold{g}_j}|^2\} +{\rho_j}\mathbb{E}\{|\hat{\bold{g}}_k^\H\bold{a}|^2\}  
\nonumber \\
&\,\,+\mathbb{E}\{|\tilde{\bold{g}}_k^\H(\sqrt{\xi_j}\hat{\bold{g}}_j +\sqrt{\rho_j}\bold{a})|^2\}\}\nonumber \\
&=\xi_jN_\mathtt{t}{{\varrho}}_k \hat{{\varrho}}_j+\rho_jN_\mathtt{t}{\varrho}_k+\xi_jN_\mathtt{t}^2\hat{\varrho}_k\hat{\varrho}_j.
\label{compute_ui_mrt}
\end{align}
Here, the last equality follows that 
$\mathbb{E}\{|\hat{\bold{g}}^\H_k\hat{\bold{g}_j}|^2\}=    \big|\mathtt{tr}\big((\hat{\bold{R}}_k^{\H})^{\frac{1}{2}}\hat{\bold{R}}_j^{\frac{1}{2}}\big)\big|^2+    \mathtt{tr}\big(\hat{\bold{R}}_k^{\H}\hat{\bold{R}}_j\big) $ \cite{uniyal_wcnc}.
Finally, by using \eqref{compute_d_mrt}, \eqref{compute_bu1_mrt}, and \eqref{compute_ui_mrt} in \eqref{eq_SINR}, we obtain \eqref{rate_final} for the MRT case.

\subsubsection{ZF beamformer} Denote by $\check{\bold{g}}_k$ the $k$-th column of the matrix $[\hat{\mathbf{G}}(\hat{\mathbf{G}}^\H\hat{\mathbf{G}})^{-1}]$, i.e.,
$\check{\bold{g}}_k \triangleq [\hat{\mathbf{G}}(\hat{\mathbf{G}}^\H\hat{\mathbf{G}})^{-1}]_{:,k}$. Thus, we have $\mathbf{f}_k^\mathtt{ZF}=\check{\mathbf{g}}_k$, based on this we derive the following.

$\bullet$ Compute $\mathbb{E}\{\bold{g}^\H_k\bold{w}^\mathtt{ZF}_k\}$:
Similar to \eqref{compute_bu_mrt}, we have
\begin{align}
\label{u_zf}
 \mathbb{E}\{\bold{g}^\H_k\bold{w}^\mathtt{ZF}_k\}&=  \mathbb{E}\{\hat{\bold{g}}^\H_k(\sqrt{\xi_k}\check{\bold{g}}_k +\sqrt{\rho_k}\bold{a}) \nonumber \\
 &\stackrel{\text{(a)}}=  \sqrt{\xi_k}\mathbb{E}\{[\hat{\mathbf{G}}^\H \hat{\mathbf{G}}(\hat{\mathbf{G}}^\H\hat{\mathbf{G}})^{-1}]_{k,k} \}= \sqrt{\xi_k},
\end{align}
where the equality $\text{(a)}$ follows that $\mathbb{E}\{\hat{\mathbf{g}}_k\}=\mathbf{0}$.

$\bullet$ Compute $\mathbb{E}\{|\bold{g}^\H_k\bold{w}^\mathtt{ZF}_k|^2\}$: Here, we have 
\begin{align} \label{bu_zf}
\mathbb{E}\{|{\bold{g}}_k\bold{w}^\mathtt{ZF}_k|^2\}
&={\xi_k}\mathbb{E}\{|\hat{\bold{g}}_k^\H \check{\bold{g}}_k|^2\} +{\rho_k}\mathbb{E}\{|\hat{\bold{g}}_k^\H\bold{a}|^2\} \nonumber \\
&\hspace{0.6cm}+{\xi_k}\mathbb{E}\{|\tilde{\bold{g}}_k^\H\check{\bold{g}}_k|^2\} +{\rho_k}\mathbb{E}\{|\tilde{\bold{g}}_k^\H{\bold{a}}|^2\}\nonumber \\
&=\xi_k+ \rho_kN_\mathtt{t}\varrho_k+\frac{\xi_k {\tilde{\varrho}_k}}{(N_\mathtt{t}-K)\hat{\varrho}_k},
\end{align}
where the last equality follows  $\varrho_k=\hat{\varrho}_k+\tilde{\varrho}_k$ and \cite{ngo}. 

$\bullet$ Compute $\mathbb{E}\{|\bold{g}^\H_k\bold{w}^\mathtt{ZF}_j|^2\}$:  Similar to \eqref{compute_bu_mrt}, we write
\begin{align}
\label{iui_zf}
    \mathbb{E}\{|\bold{g}^\H_k\bold{w}^\mathtt{ZF}_j|^2\}&={\xi_j}\mathbb{E}\{|\hat{\bold{g}}^\H_k\check{\bold{g}}_j|^2\} +{\rho_j}\mathbb{E}\{|\hat{\bold{g}}_k^\H\bold{a}|^2\}  
\nonumber \\
&\,\,+\mathbb{E}\{|\tilde{\bold{g}}_k^\H(\sqrt{\xi_j}\check{\bold{g}}_j +\sqrt{\rho_j}\bold{a})|^2\}\}\nonumber \\
&=\rho_jN_\mathtt{t}\varrho_k+\frac{\xi_j {\tilde{\varrho}_k}}{(N_\mathtt{t}-K)\hat{\varrho}_j},
\end{align}
where last equality follows that  $\mathbb{E}\big\{ \big|\hat{\mathbf{g}}_{k}^{\H}\check{\mathbf{g}}_{j}\big|^{2} \big\}
= \mathbb{E}\big\{ \big[\hat{\mathbf{G}}^{\H}\hat{\mathbf{G}}(\hat{\mathbf{G}}^{\H}\hat{\mathbf{G}})^{-1}\,\big]_{k,j} \big\}
\!\!\!\!\!=\!\!\!\! 0$.
Finally, using \eqref{u_zf}--\eqref{iui_zf} in \eqref{eq_SINR} along  with some algebraic manipulations yields \eqref{rate_final}. 

\section{proof of Remark~\ref{rem:sum_rate_scaling}}\label{sec:rate_scaling_law}
\renewcommand{\theequation}{C.\arabic{equation}}
\renewcommand{\theHequation}{C.\arabic{equation}}
\setcounter{equation}{0}

For equal power allocation, from \eqref{eq:total_transmit_power}, we have $\eta= \frac{P_\mathtt{tot}^{\mathtt{bf}}}{2N_\mathtt{t}}$ and $\xi_k = \frac{P_\mathtt{tot}^{\mathtt{bf}}}{2N_\mathtt{t} \sum_{j=1}^{K} \nu^{\mathtt{bf}}_j}, \forall k$. By substituting these 
$\eta$ and $\xi_k$ into  \eqref{rate_final}, we obtain
$\bar{\mathcal R}^{\mathtt{bf}}_{k} =\tau_\mathtt{0}\log_2\left(1+{\bar{\varsigma}^{\mathtt{bf}}}_k\right)$, 
where 
\begin{align}
\label{app:equal_power_sinr}
\hspace{-0.2cm}{\bar{\varsigma}_k^{\mathtt{bf}}}\!=\!
\begin{cases}
\displaystyle
\frac{N_\mathtt{t}\,\hat{\varrho}_k^{2}\,P_\mathtt{tot}^{\mathtt{bf}}}
{2\left(\varrho_k P_\mathtt{tot}^{\mathtt{bf}}\!+\!\sigma_\mathtt{z}^{2}\right)\sum_{j=1}^{K}\hat{\varrho}_j}
, & \hspace{-0.3cm}\text{for MRT },\\
\displaystyle
\frac{(N_\mathtt{t}-K)P_\mathtt{tot}^{\mathtt{bf}}}
{\left((\varrho_k\!+\!\tilde \varrho_k)P_\mathtt{tot}^{\mathtt{bf}}+2\sigma_\mathtt{
z}^{2}\right)\sum_{j=1}^{K}\hat{\varrho}_j^{-1}}
, & \hspace{-0.3cm}\text{for ZF }.
\end{cases}
\end{align}
We next derive the large $N$ approximation for $\tilde{\varrho}_k$ and $\hat{\varrho}_k$. 
Recall from Remark~\eqref{rem:nmse_scaling} that $\varrho_k\to \infty $ as $N \to \infty$. Using this along with \eqref{rk_hat} into $\tilde{\varrho}_k\triangleq \varrho_k -\hat{\varrho}_k$, we get $\tilde{\varrho}_k \approx N \tilde{\varkappa}_k+\frac{\sigma^2_\mathtt{z}/(\tau_\mathtt{p}p_\mathtt{p})}{1+\frac{\sum_{j \in \mathcal{P}_k \backslash \{k\} } \beta_{\mathtt{UR}_j}}{\beta_{\mathtt{UR}_k}}}$, where $\tilde{\varkappa}_k=\frac{\sum_{j \in \mathcal{P}_k \backslash \{k\} } \varkappa_k}{1+ \frac{\sum_{j \in \mathcal{P}_k \backslash \{k\} }\beta_{\mathtt{UR}_j}}{\beta_{\mathtt{UR}_k}}}$. 
Here, with pilot contamination and as $N\to\infty$, we have $\tilde \varrho_k \approx N\tilde \varkappa_k$; without pilot contamination, we have $\tilde \varrho_k \approx \sigma^2_{\mathtt z}/(\tau_{\mathtt p}p_{\mathtt p})$.
Thus, we can write $\hat{\varrho}_k\approx N(\varkappa_k-\tilde{\varkappa}_k)$ with pilot contamination, and $\hat{\varrho}_k\approx N\varkappa_k$ without pilot contamination. Substituting the approximations for $\varrho_k$, $\tilde{\varrho}_k$, and $\hat{\varrho}_k$, and $P_\mathtt{tot}^{\mathtt{bf}}=e_\mathtt{tot}/N^\varepsilon$ into \eqref{app:equal_power_sinr} 
and 
retaining only the dominant $N$ terms,
we get \eqref{sinr_approx}.


\section{proof of Theorem~\ref{thm:CRLB}}\label{sec:appendixe}
\renewcommand{\theequation}{D.\arabic{equation}}
\renewcommand{\theHequation}{D.\arabic{equation}}
\setcounter{equation}{0}
We first compute the closed-form expressions for the FIM entries in
\eqref{eq:tpp}--\eqref{eq:tbb}. Substituting $\mathbf{R}_{\mathbf{x}}^{\mathtt{bf}}$ from
Appendix~\ref{sec:transmit_covaraince}, with
$\mathbf{u}=\mathbf{a}(\theta,\phi)$, into
\eqref{eq:tpp}--\eqref{eq:tbb}, and applying the trace identities
together with algebraic simplifications based on the fact that
$\mathbf{a}^{\H}\dot{\mathbf{a}}_{\theta}
=\mathbf{a}^{\H}\dot{\mathbf{a}}_{\phi}=0$, we obtain
\begin{align}
\label{app:F_psi_psi}
&T_{\psi\psi} \!\!= \!\!|{{{\beta}}_\mathtt{s}}|^{2} 
\!\!\left(\!
\boldsymbol{\xi}^{\T}\!\boldsymbol{\nu}^{\mathtt{bf}} \!\! 
\left( \!\!N_\mathtt{r}\|\dot{\mathbf{a}}_{\psi}\!\|^{2} \!\!+ \!\!N_\mathtt{t}\|\dot{\mathbf{b}}_{\psi}\!\|^{2}\! \right)\!\!+ \!\eta N_\mathtt{t}^{2}\|\dot{\mathbf{b}}_{\psi}\!\|^{2}\!
\right)\!,  \\
\label{app:F_theta_phi}
&T_{\theta\phi} \!\!=\!\! |{{{\beta}}_\mathtt{s}}|^{2}\!\!
\left(
\boldsymbol{\xi}^{\T}\!\boldsymbol{\nu}^{\mathtt{bf}}\!\!
\left( N_\mathtt{r}\dot{\mathbf{a}}_{\theta}^{\H}\dot{\mathbf{a}}_{\phi}
     \!\!+\!\! N_\mathtt{t}\dot{\mathbf{b}}_{\theta}^{\H}\dot{\mathbf{b}}_{\phi} \right)\!\!
+ \!\eta N_\mathtt{t}^{2}\dot{\mathbf{b}}_{\theta}^{\H}\dot{\mathbf{b}}_{\phi}\!
\right), \\
&\mathbf{T}_{\tilde{\bm{\beta}_\mathtt{s}}\tilde{\bm{\beta}}_\mathtt{s}}
= \left(\boldsymbol{\xi}^{\T}\boldsymbol{\nu}^{\mathtt{bf}}N_\mathtt{t}N_\mathtt{r}
+ \eta N_\mathtt{t}^{2}N_\mathtt{r}\right)\mathbf{I}_{2},\\
&\mathbf{t}_{\theta\tilde{\boldsymbol{\beta}}_\mathtt{s}}
=
\mathbf{t}_{\phi\tilde{\boldsymbol{\beta}}_\mathtt{s}}
=
\mathbf{0}_{1\times 2}.
\label{app:F_phi_beta}
\end{align}
{Using \eqref{app:F_phi_beta}} in
\eqref{eq:tilde_T_theta_theta}--\eqref{eq:tilde_T_phi_phi},
we obtain
\begin{align}
\widetilde{T}_{\theta\theta}
=
\widetilde{T}_{\theta\phi}
=
\widetilde{T}_{\phi\phi}
=
0.
\label{app:tilde_T_zero}
\end{align}
 Then, to compute the terms $\|\dot{\mathbf{a}}_{\psi}\|^{2}$ and $\dot{\mathbf{a}}_{\theta}^{\H}\dot{\mathbf{a}}_{\phi}$ in \eqref{app:F_psi_psi}, we apply $\partial(\mathbf{a} \otimes \mathbf{b})/\partial_{x}=\dot{\mathbf{a}}_x\otimes \mathbf{b} +\mathbf{a} \otimes \dot{\mathbf{b}}_x$ in \eqref{steer_vector} to get
\begin{align}
   \dot{\mathbf{a}}_{\theta}
&= \dot{\mathbf{a}}_{\mathtt{y}\theta} \otimes \mathbf{a}_{\mathtt{z}}
+ \mathbf{a}_{\mathtt{y}} \otimes \dot{\mathbf{a}}_{\mathtt{z}\theta}
= \dot{\mathbf{a}}_{\mathtt{y}\theta} \otimes \mathbf{a}_{\mathtt{z}},\label{app:a_dot_theta}
\\ \label{app:a_dot_phi}\dot{\mathbf{a}}_{\phi}
&= \dot{\mathbf{a}}_{\mathtt{y}\phi} \otimes \mathbf{a}_{\mathtt{z}}
+ \mathbf{a}_{\mathtt{y}} \otimes \dot{\mathbf{a}}_{\mathtt{z}\phi}.
\end{align}
Here,  $\dot{\mathbf{a}}_{\mathtt{z}\theta}=0$ since $\mathbf{a}_\mathtt{z}$ is independent of $\theta$ and
\begin{align}
&\!\!\!\dot{\mathbf{a}}_{\mathtt{y}\theta}
\!=\!\! j\pi \cos(\theta)\sin(\phi)\,\mathbf{v}_{\mathtt{ty}} \circ \mathbf{a}_{\mathtt{y}}, 
\label{app:a_dot_y_theta}\\
\label{app:a_dot_y_phi}
&\!\!\!\dot{\mathbf{a}}_{\mathtt{y}\phi}
\!= \!\!j\pi \!\cos(\phi)\!\sin(\theta)\mathbf{v}_{\mathtt{ty}} \!\circ \!\mathbf{a}_{\mathtt{y}}, \dot{\mathbf{a}}_{\mathtt{z}\phi}
\!\! = \!\!-j\pi \! \sin(\phi)\mathbf{v}_{\mathtt{tz}} \!\circ\! \mathbf{a}_{\mathtt{z}},
\end{align} where $\mathbf{v}_\mathtt{u}=[-(N_\mathtt{u}-1)/2,\cdots, (N_\mathtt{u}-1)/2]$, $\mathtt{u} \in \{\mathtt{ty}, \mathtt{tz},\mathtt{ry},\mathtt{rz}\}$, and
$\|\mathbf{v}_\mathtt{u}\|^2=N_\mathtt{u}(N_\mathtt{u}^2-1)/12$.
Furthermore, from \eqref{steer_vector}, \eqref{app:a_dot_theta}-\eqref{app:a_dot_y_phi}, we have $\mathbf{a}^{\H}\dot{\mathbf{a}}_{\theta} 
= (\mathbf{a}_{\mathtt{y}}^{\H}\dot{\mathbf{a}}_{\mathtt{y}\theta})
  \otimes (\mathbf{a}_{\mathtt{z}}^{\H}\mathbf{a}_{\mathtt{z}})
\overset{\text{(a)}}{=}  0, 
\mathbf{a}^{\H}\dot{\mathbf{a}}_{\phi} 
= (\mathbf{a}_{\mathtt{y}}^{\H}\dot{\mathbf{a}}_{\mathtt{y}\phi}
   \otimes \mathbf{a}_{\mathtt{z}}^{\H}\mathbf{a}_{\mathtt{z}})
 + (\mathbf{a}_{\mathtt{y}}^{\H}\mathbf{a}_{\mathtt{y}}
   \otimes \mathbf{a}_{\mathtt{z}}^{\H}\dot{\mathbf{a}}_{\mathtt{z}\phi})
\overset{\text{(a)}}{=}  0$.
Here, the equality $\text{(a)}$ follows that $\mathbf{a}_{\mathtt{y}}^{\H}\dot{\mathbf{a}}_{\mathtt{y}\theta}
= \mathbf{a}_{\mathtt{y}}^{\H}\dot{\mathbf{a}}_{\mathtt{y}\phi} = 0$. Note that the receive steering vector $\mathbf{b}$ has a structure analogous to $\mathbf{a}$. Thus, the corresponding terms are derived similarly and omitted for brevity.
From \eqref{app:a_dot_theta}--\eqref{app:a_dot_y_phi} with $N_\mathtt{t}=N_\mathtt{ty}N_\mathtt{tz}$ and $\mathbf{a}_{\mathtt{y}}^{\H}\dot{\mathbf{a}}_{\mathtt{y}\theta}
= \mathbf{a}_{\mathtt{y}}^{\H}\dot{\mathbf{a}}_{\mathtt{y}\phi}
= \mathbf{a}_{\mathtt{z}}^{\H}\dot{\mathbf{a}}_{\mathtt{z}\theta}
= \mathbf{a}_{\mathtt{z}}^{\H}\dot{\mathbf{a}}_{\mathtt{z}\phi}
= 0$, we have
$\|\dot{\mathbf{m}}_{\theta}\|^{2}
\!\! =\! \! \frac{N_\mathtt{w}\left(N_{\mathtt{wy}}^{2}-1\right)}{12}
\pi^{2}\cos^{2}(\theta)\sin^{2}(\phi), 
\|\dot{\mathbf{m}}_{\phi}\|^{2}
\! \! = \! \! \frac{N_\mathtt{w}}{12}\pi^{2}\!\cos^{2}(\phi)\!
\left( \left(N_{\mathtt{wy}}^{2}-1\right)\!\sin^{2}(\theta)
+\!\left(\!N_{\mathtt{wz}}^{2}\!-\!1\right) \right),   
\dot{\mathbf{m}}_{\theta}^{\H}\dot{\mathbf{m}}_{\phi}
 = \frac{N_\mathtt{w}\left(N_{\mathtt{wy}}^{2}-1\right)}{48}
\pi^{2}\sin(2\phi)\sin(2\theta)$.
Here, $\mathbf{m}=\mathbf{a}$ if $\mathtt{w}=\mathtt{t}$ and $\mathbf{m}=\mathbf{b}$ if $\mathtt{w}=\mathtt{r}$.
Finally, using \eqref{app:F_psi_psi}--\eqref{app:tilde_T_zero} and $\|\dot{\mathbf{m}}_{\theta}\|^{2}$, $\|\dot{\mathbf{m}}_{\phi}\|^{2}$, and  $\dot{\mathbf{m}}_{\theta}^{\H}\dot{\mathbf{m}}_{\phi}$
in the expressions for $\mathtt{CRLB}_{\theta}$ and $\mathtt{CRLB}_{\phi} $  { in \eqref{eq:general_crlb_theta} and \eqref{eq:general_crlb_phi}}, 
along with some algebraic manipulations, we obtain \eqref{CRLB_psi}.


\section{proof of Lemma~\ref{lemma_FP}}\label{sec:appendixf}
\renewcommand{\theequation}{E.\arabic{equation}}
\renewcommand{\theHequation}{E.\arabic{equation}}
\setcounter{equation}{0}

We apply the Lagrangian dual transform \cite{shen_FP_part2} to decouple the logarithmic objective function 
in \eqref{problem1}. By introducing a variable $t^\mathtt{bf}_{\mathtt{u}k}$ to represent $\varsigma_k^{\mathtt{bf}}$, 
$\mathcal{R}_k^{\mathtt{bf}}$ in \eqref{rate_final} can be expressed as 
\begin{align}
\label{lower_bound_FP}g^\mathtt{bf}_k(\check{\bm{\alpha}},t^\mathtt{bf}_{\mathtt{u}k})\triangleq  \log(1+t^\mathtt{bf}_{\mathtt{u}k})-t_{\mathtt{u}k}+\frac{(1+t^\mathtt{bf}_{\mathtt{u}k})\varsigma_k^\mathtt{bf}}{(1+\varsigma^\mathtt{bf}_k)},
\end{align}
where $ {g}^\mathtt
{bf}_k(\check{\bm{\alpha}},t^\mathtt{bf}_{\mathtt{u}k})$ is a lower bound of the original rate $\mathcal{R}^{\mathtt{bf}}_k$, i.e., ${g}_k^\mathtt{bf}(\check{\bm{\alpha}},t^\mathtt{bf}_{\mathtt{u}k})\leq \mathcal{R}^{\mathtt{bf}}_k$. Here, the equality holds when $\partial {g}_k^\mathtt{bf}(\check{\bm{\alpha}},t^\mathtt{bf}_{\mathtt{u}k})/\partial t^\mathtt{bf}_{\mathtt{u}k}=0$, $\forall k$, yielding the optimal $t_{\mathtt{u}k}^{\mathtt{bf}^\star}$ in \eqref{aux:t_u}.
Then, we apply the FP technique
named quadratic transform \cite{shen_FP} to decouple the fractional term in \eqref{lower_bound_FP}. By introducing an auxiliary variable $t_{\mathtt{v}k}$, \eqref{lower_bound_FP} can be recast as
\begin{align}
 &\tilde{g}^\mathtt{
 bf}_k\!(\check{\bm{\alpha}},t^\mathtt{bf}_{\mathtt{u}k},t^\mathtt{bf}_{\mathtt{v}k})\!\!\triangleq\! \log(1\!+\!t^\mathtt{bf}_{\mathtt{u}k})\!-\!t^\mathtt{bf}_{\mathtt{u}k}\!+\!2t^\mathtt{bf}_{\mathtt{v}k}\sqrt{(1\!+\!t^\mathtt{bf}_{\mathtt{u}k})\xi_k{u_k^{\mathtt{bf}}\!(\check{\bm{\alpha}})}} \nonumber \\
 &- \!(t^\mathtt{bf}_{\mathtt{v}k})^2\!\bigg(\!\!\eta d_k (\check{\bm{\alpha}}) \!+\!\sigma_{\mathtt{z}}^2(\check{\bm{\alpha}})\!+\!{{ \sum _{j=1}^{K} \!\xi_j c_{kj}^{\mathtt{bf}} \!(\check{\bm{\alpha}})
 }\!+\!{u_k^{\mathtt{bf}}(\check{\bm{\alpha}})\xi_k}}\!\!\bigg),
 \label{f_tilde}
\end{align}
where $d_k(\check{\bm{\alpha}})$ $c_{kj}^\mathtt{bf}(\check{\bm{\alpha}})$ and $u_k^\mathtt{bf}(\check{\bm{\alpha}})$ are given in \eqref{sinr_final}.
Besides, $\tilde{g}^\mathtt{
 bf}_k(\bm{\alpha},t^\mathtt{bf}_{\mathtt{u}k},t^\mathtt{bf}_{\mathtt{v}k})\leq {g}^\mathtt{
 bf}_k(\bm{\alpha},t^\mathtt{bf}_{\mathtt{u}k})$ 
 and the equality holds when $\partial \tilde{g}^\mathtt{
 bf}_k(\bm{\alpha},t^\mathtt{bf}_{\mathtt{u}k},t^\mathtt{bf}_{\mathtt{v}k})/\partial t^\mathtt{bf}_{\mathtt{v}k}=0$, $\forall k$, yielding the optimal $t_{\mathtt{v}k}^{\mathtt{bf}^\star}$ in \eqref{aux:t_v} \cite[Lemma~1]{shen_FP}. Note that in \eqref{f_tilde}, the term $c_{kj}^\mathtt{bf}(\check{\bm{\alpha}})$ contains a fractional term $\hat{\varrho}_k(\check{\bm{\alpha}})$.
To tackle this, we apply quadratic transform to $\hat{\varrho}_k(\check{\bm{\alpha}})$, similar to \eqref{lower_bound_FP}. 
By introducing an auxiliary variable $t_{\mathtt{w}k}$, a linear lower bound of $\hat{\varrho}_k(\check{\bm{\alpha}})$, denoted by $\check{\varrho}_k(\check{\bm{\alpha}},t_{\mathtt{w}k})$ can be achieved,  and is given in \eqref{r_k_check}. More specifically, we have $\hat{\varrho}_k(\check{\bm{\alpha}})\geq\check{\varrho}_k(\check{\bm{\alpha}},t_{\mathtt{w}k})$, where the equality is achieved when 
$\partial \check{\varrho}_k(\check{\bm{\alpha}},t_{\mathtt{w}k})/\partial t_{\mathtt{w}k}=0$, $\forall k$, yielding the optimal $t_{\mathtt{w}k}^\star$ in \eqref{aux:t_w}.
Finally, substituting $\check{\varrho}_k$ into \eqref{f_tilde} along with some algebraic manipulations yields the objective function $f^\mathtt{bf}_{k}(\check{\bm{\alpha}},t^\mathtt{bf}_{\mathtt{u}k}, t^\mathtt{bf}_{\mathtt{v}k},t_{\mathtt{w}k})$ in
\eqref{lem:objecive_hris}. This concludes the proof.

\vspace{-0.1cm}

\bibliographystyle{IEEEtran}
\bibliography{IEEEabrv,reference}

\end{document}

%% file: Definitions.tex
\usepackage{amsmath,graphicx}
\usepackage{bm}
\usepackage{color}
\usepackage{graphicx}
\usepackage{epstopdf}
\usepackage{amsmath}
\usepackage{amssymb}
\usepackage{mathrsfs}
\usepackage[colorlinks,
            linkcolor=blue,
            anchorcolor=red,
            citecolor=red,
            urlcolor=black]{hyperref}
\usepackage{tikz}
\usepackage{bm}
\usepackage[english]{babel}
\usepackage{cite}
\usepackage{rotfloat}
\usepackage{mathtools}
\usepackage[font=normalsize,labelfont=bf]{caption}
\usepackage{amsmath}
\usepackage{makecell}
\usepackage{multirow}
\usepackage{subfigure}
\usepackage{booktabs}
\usepackage{colortbl}
\usepackage{multirow}
\usepackage{hhline}
\usepackage{stfloats}
\usepackage{multicol}
\usepackage{bbm}
\usepackage{cases}
\graphicspath{ {Figures/} }
\newcommand{\T}{{\scriptscriptstyle\mathsf{T}}}
\renewcommand{\H}{{\scriptscriptstyle\mathsf{H}}}

\newsavebox{\foobox}

\definecolor{kugray5}{RGB}{224,224,224}

\usepackage[normalem]{ulem}
\newcommand\rsout{\bgroup\markoverwith
	{\textcolor{red}{\rule[0.5ex]{2pt}{0.8pt}}}\ULon}



\newcommand{\blue}[1]{\textcolor{blue}{#1}}

\makeatletter
\newcommand{\ALOOP}[1]{\ALC@it\algorithmicloop\ #1%
	\begin{ALC@loop}}
	\newcommand{\ENDALOOP}{\end{ALC@loop}\ALC@it\algorithmicendloop}

\makeatother

\usepackage{etoolbox}
\let\mybibitem\bibitem
\renewcommand{\bibitem}[1]{%
	\ifstrequal{#1}{nature}
	{\color{blue}\mybibitem{#1}}
	{\color{black}\mybibitem{#1}}%
}

\graphicspath{ {Figures/} }

\newtheorem{theorem}{\text{Theorem}}
\newtheorem{remark}{\text{Remark}}
\newtheorem{lemma}{\text{Lemma}}

\newtheorem{proof}{Proof}

\newcommand{\qed}{\ensuremath{\square}}  

\DeclareCaptionLabelSeparator{periodspace}{.\quad}

\addto\captionsenglish{}
\allowdisplaybreaks

\usepackage{setspace}

\newcommand{\mR}{{\mathbf{R}}}

\newcommand{\mA}{{\mathbf{A}}}

\newcommand{\mI}{\textbf{\textbf{I}}}
\newcommand{\mT}{{\mathbf{T}}}

\newcommand{\vu}{{\mathbf{u}}}

\newcommand{\vb}{{\mathbf{b}}}

\newcommand{\va}{{\mathbf{a}}}

\newcommand{\vt}{{\mathbf{t}}}

\def\b0{{\pmb{0}}}

\newcommand{\Jtt}{T_{\theta \theta}}
\newcommand{\Jtp}{T_{\theta \phi}}
\newcommand{\Jta}{\vt_{\theta \tilde{\bm{\beta_\mathtt{s}}}}}
\newcommand{\Jpp}{T_{\phi \phi}}
\newcommand{\Jpa}{\vt_{\phi \tilde{\bm{\beta_\mathtt{s}}}}}
\newcommand{\Jaa}{\mT_{\tilde{\bm{\beta_\mathtt{s}}}\tilde{\bm{\beta_\mathtt{s}}}}}

\newcommand{\adotphi}{\dot{\va}_{\phi}}
\newcommand{\adottheta}{\dot{\va}_{\theta}}
\newcommand{\bdotphi}{\dot{\vb}_{\phi}}
\newcommand{\bdottheta}{\dot{\vb}_{\theta}}